\documentclass[runningheads]{llncs}

\pdfoutput=1

\usepackage{cite}
\usepackage{amsmath,amssymb,amsfonts}
\usepackage{mathdots}
\usepackage{algorithmicx}
\usepackage{graphicx}
\usepackage{subcaption}
\usepackage{textcomp}
\usepackage[dvipsnames]{xcolor}
\usepackage{orcidlink}
\usepackage{hyperref}
\usepackage{todonotes}
\usepackage{pgfplots}
\usepackage{pgf}
\usepackage{graphicx}
\usepackage{svg}
\usepackage{multirow}
\usepackage{bbm}
\usepackage{tikz}
\usetikzlibrary{backgrounds}
\usetikzlibrary{quantikz2}
\usetikzlibrary{shapes.geometric}

\usepackage[english]{babel}

\usepackage{algorithm2e}
\RestyleAlgo{ruled}
\usepackage[noend]{algpseudocode}
\usepackage{dsfont}
\usepackage{listings}

\usepackage[capitalise]{cleveref}
\usepackage{booktabs}
\usepackage{diagbox}

\usepackage{tabularx}
\usepackage{array}

\usepackage{chngcntr}
\usepackage{booktabs}

\usepackage[T1]{fontenc}

\usepackage{graphicx}

\pgfplotsset{compat=1.18}
\begin{document}
\title{GPU-Accelerated Host-Aware Dead-Measurement Detection in Hybrid Quantum--Classical Programs: Full Version}

\titlerunning{Host-Aware Dead-Measurement Detection}

\author{
Yanbin~Chen\inst{1}\orcidID{0000-0002-1123-1432} \and
Qunyou~Liu\inst{2}\orcidID{0000-0002-7410-502X} \and
Yu~Wang\inst{1}\orcidID{0009-0004-3972-4388} \and
Christian~B.~Mendl\inst{1}\orcidID{0000-0002-6386-0230} \and
Helmut~Seidl\inst{1}\orcidID{0000-0002-2135-1593}
}

\authorrunning{Y. Chen et al.}

\institute{
TUM School of CIT,
Technical University of Munich,\\
Boltzmannstr.~3,
85748 Garching, Germany\\
\email{\{yanbin.chen,18yu.wang,christian.mendl,helmut.seidl\}@tum.de}
\and
School of Engineering,
\'{E}cole Polytechnique F\'{e}d\'{e}rale de Lausanne (EPFL),\\
CH-1015 Lausanne, Switzerland\\
\email{qunyou.liu@epfl.ch}
}

\maketitle             
\begin{abstract}

Hybrid programs combine a quantum circuit with a classical host program that consumes measurement outcomes. In such programs, an outcome may be syntactically read by the host but
semantically non-contributory: changing the outcome cannot change the returned value. Such outcomes obscure gates that are dead only relative to the host semantics, and are therefore invisible to circuit-local optimizers.

We present a semantics-aware host-side static analysis that identifies non-contributory measurement outcomes by abstract interpretation, and prove its soundness. We implement the analysis and evaluate it on $24$ application-faithful hybrid workloads across quantum chemistry, optimization, quantum machine learning, and quantum finance.
Compared with a syntactic liveness baseline, our analysis identifies more than $4\times$ as many non-contributory measurements, and it standalone enables the removal of $37.98\%$ of total gates on average. Even after the state-of-the-art optimizers like Qiskit, t|ket$\rangle$, and PyZX have already optimized the circuits, our analysis still enables removal of more than $30\%$ of the post-optimized gates, showing that the host-semantic opportunities exposed by our analysis are not subsumed by circuit-local optimization.
To scale our analysis, we further lower host programs to an SSA-style levelized intermediate representation that exposes level-wise parallelism for GPU execution, and implement a CUDA backend. We prove that this lowering preserves the analysis result, and the evaluation shows speedups of up to $6.53\times$ over a sequential baseline as structural parallelism increases.

\end{abstract}

\SetArgSty{textnormal}

\section{Introduction}\label{sec:intro}
Hybrid quantum--classical programs execute a quantum circuit whose measurement outcomes are consumed by a classical host program to compute the final result \cite{Preskill_2018, 8638598, bouland2020prospectschallengesquantumfinance, 10313603, 10313887, quetschlich2024equivalencecheckingclassicalcircuits, jojo2024quantumalgorithmstensorsvd}.
Although every measured qubit is typically read and its value flows through the host code, not all of those values truly influence the final return of the program.
Contributions of these values could be semantically neutralized. 
As exemplified by \cref{ex:motivating-example-classical-var-dead}, a value may cancel algebraically with others, vanish under an $\texttt{int}(\cdot)$ truncation, or appear inside a sub-expression whose effect is later overwritten.

\begin{example}\label{ex:motivating-example-classical-var-dead}
    In \cref{fig:motivating-example-hybrid-program} the quantum circuit produces outcomes $o_0,o_1,o_2$ which the host binds to $\texttt{a,b,c}$. Although $\texttt{a}$ is read multiple times and even guards a branch, the returned value is independent of the initial value feeding $a$. If the then-branch is taken, then $\texttt{u+v+w}=(\texttt{a+b})(\texttt{c-a})+\texttt{a}(\texttt{b-c})+\texttt{a}^2 = \texttt{bc}$, where all $\texttt{a}$–terms cancel. If the else-branch is taken, with $\eta\!\in\!(0.1,0.5)$ and $\texttt{a,b,c}\!\in\!\{0,1\}$, we have $\texttt{int}(\eta\,\texttt{a+bc})=\texttt{bc}$, and $\texttt{u+v+w}=\texttt{bc}$. Thus, despite being syntactically used, the initial value originating from $o_0$ is non‑contributory, i.e., changing $o_0$ does not influence the semantics of the program.

\end{example}

\begin{remark}
    In \cref{ex:motivating-example-classical-var-dead},
    if the else-branch is taken, then semantically \(\texttt{a}=0\) and we get $\texttt{int}(\eta\,\texttt{a+bc})=\texttt{int}(\texttt{bc})=\texttt{bc}$ without even accessing the range of $\eta$. However, this would need a mechanism of propagating the branch-local constraint into branches, which we have not implemented within our analysis.
\end{remark}

The standard liveness analysis \cite{kildall1973unified, alfred2007compilers} is not enough to detect the semantic deadness of $\texttt{a}$ in \cref{fig:motivating-example-hybrid-program}, as the liveness analysis is syntactic-based: $\texttt{a}$ appears in the branch condition \texttt{if(a)}, in the right-hand side of assignments in both branches, all of which syntactically contribute to the returned variables \texttt{u}, \texttt{v}, \texttt{w}, so \texttt{a} must be classified as potentially live.

\begin{figure}[htb]
    \centering
\begin{tikzpicture}[>=latex,node distance=1.5em]
  ... 

 \node[](q2)
 {
    \colorbox{blue!15}{
        \begin{minipage}{0.4\textwidth}
\begin{quantikz}[row sep=0.2cm, column sep=0.5cm]
        &\targ{}\gategroup[2,steps=1,style={draw=red, dashed, rounded corners, inner sep=0.00cm}]{}   &\gate{V_1}\gategroup[3,steps=1,style={draw=red, dashed, rounded corners, inner sep=0.00cm}, label style={text=red}]{Redundant}&\gate{V_2}\gategroup[3,steps=1,style={draw=red, dashed, rounded corners, inner sep=0.00cm}]{}&\meter{}\rstick{$o_0 \in \{0, 1\}$}\\
      &\ctrl{-1} &\ctrl{-1} &          &\meter{}\rstick{$o_1 \in \{0, 1\}$} \\
      &\gate{W_1}&\ctrl{-2} &\ctrl{-2} &\meter{}\rstick{$o_2 \in \{0, 1\}$}
\end{quantikz}
      \end{minipage}
    }
 }; 
 \node[below= 0.25cm of q2](c3)
{
    \colorbox{pink}{
        \begin{minipage}{0.45\textwidth}
            \SetAlgoLined
            \SetNlSty{textbf}{}{:}
            \begin{algorithmic}
            
              \State \textbf{if}(\texttt{a}):
              \State $\quad \texttt{x, y} \gets \texttt{a + b, c - a}$
              \State $\quad \texttt{u,v,w} \gets \texttt{x} \cdot \texttt{y}$, $\texttt{a} \cdot (\texttt{b - c})$, $\texttt{a}^2$

              \State \textbf{else}:
              \State $\quad\eta \gets \textbf{random}(0.1, 0.5)$
              \State $\quad \texttt{u,v,w}\gets \textbf{int}(\eta \texttt{a+bc})$, $\texttt{b-c, c-b}$
              \State \Return $\texttt{u + v + w}$
              
            \end{algorithmic}
      \end{minipage}
    }
};

\draw[red, dashed, rounded corners, very thick]
       (0.95, 1.0) rectangle (2.7, 0.1);
\node[red] at (1.7,1.3) {$\textbf{non-contributory}$};
    
 \path[->]
(q2) edge [] node [right] {$\texttt{a, b, c} \gets o_0, o_1, o_2$} (c3);
  \begin{scope}[on background layer]
    \node[
      fit=(q2) (c3),
      inner sep=1pt,
      rounded corners=3pt,
      fill=yellow!15
    ] {};
  \end{scope}
\end{tikzpicture}
    \caption{A hybrid program where the host variable $\texttt a$, bound to measurement outcome $o_0$, is syntactically used but semantically non-contributory: the program’s return is invariant under changes to $o_0$. By non-contribution of $o_0$, the quantum gates enclosed by red dashed boxes are identified as redundancies: their removal leaves the host program’s returned result unchanged.}
    \label{fig:motivating-example-hybrid-program}
\end{figure}
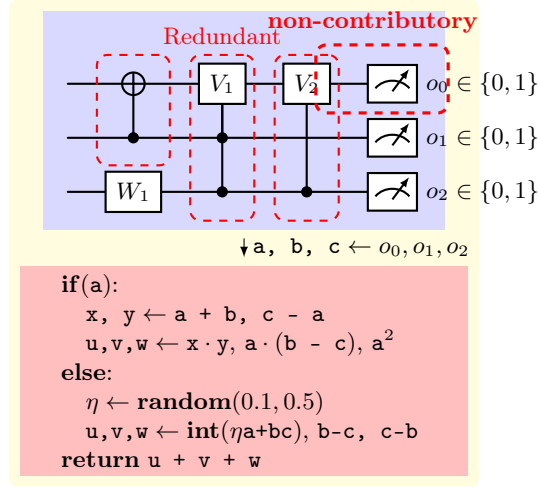

From the hardware perspective, measuring qubits whose outcomes are later neutralized increases readout traffic \cite{khammassi2022scalable, van2019electronic} and adds readout error \cite{maciejewski2020mitigation, nachman2020unfolding} without any benefit to the final result. Performing gates that only influence those qubits is likewise unnecessary: they deepen the circuit, extend the required coherence window, and accumulate error across extra operations, which lowers overall fidelity despite contributing nothing to the result of calculation. For example, in the hybrid program \cref{fig:motivating-example-hybrid-program}, given $o_0$ (feeding $\texttt{a}$) being non-contributory, the gates highlighted by the red dashed boxes are therefore dead and can be removed without changing the distribution of the contributory outcomes, i.e., $o_1$ and $o_2$, or the host program’s returned result.

As quantum stacks mature, cross‑disciplinary teams prototype quickly across library abstractions and mixed toolchains \cite{broughton2021tensorflowquantumsoftwareframework, bergholm2022pennylaneautomaticdifferentiationhybrid, Steiger2018projectqopensource, LaRose2022mitiqsoftware}, and small but regular host‑side edits that obscure data flow could turn a seemingly used measurement into a non‑contributory one, of which effects are brittle under iteration and hard to catch in review, but the cost of leaving them in is paid on every shot of execution.
Therefore, identifying as a compilation step semantically non-contributory measurement outcomes, and removing redundant operations accordingly, is crucial to improving quantum circuit efficiency and reliability.

\noindent\textbf{Prior works.}
Existing methods address this problem only in part.
Dead gate elimination (DGE) provides a rigorous circuit-level framework: if the action of a gate does not affect the probability distribution of contributory measurement outcomes, it can be removed without altering observable behaviour \cite{10.1007/978-3-031-97632-2_10}, but DGE assumes that the set of non-contributory outcomes is already known and cannot infer which ones semantically contribute to the host’s return value.
In general, deciding semantic liveness at a program point---whether there exists an execution whose observable behaviour depends on the current value of a variable---is undecidable \cite{Péter_1954,spa}, so classical liveness analyses are conservative static approximations computed via backward data-flow or abstract interpretation \cite{bookFlemming1999,seidl2012compiler}.
Remme et\,al.\ optimize hybrid Quil programs using standard compiler analyses, but their liveness is a purely syntactic is-read notion and cannot establish result-invariance of measurement outcomes \cite{remme2025optimizationhybridquantumclassicalalgorithms}.
Static analyzers such as LintQ and QChecker provide AST-driven checks for suspicious Qiskit patterns, yet focus on syntactic use rather than host-semantic non-contribution and do not drive circuit pruning \cite{zhao2023qcheckerdetectingbugsquantum,Paltenghi_2024}.
Partial equivalence checking can certify that a simplified circuit preserves a user-chosen subset of outputs, but does not identify which outputs are irrelevant nor perform simplification automatically \cite{Chen_2022_Partial_Eq_Check}; similarly, \emph{QuTracer} prunes gates that do not affect a user-specified subset of measured qubits, but cannot discover that subset in a hybrid setting \cite{Li_2024}.
Extending code property graphs to quantum code enables cross-domain analysis and can detect measurements whose results are never used \cite{Kaul_2023}, but it cannot prove that syntactically used results are nonetheless semantically non-contributory.
Quantum constant propagation exploits contextual host information to fold gates and simplify unitary blocks \cite{chen_QCP_2023}; however, information flows only from host to circuit and does not track how measurement results are subsequently consumed, so it cannot prove semantic non-contribution of syntactically used outcomes.
Finally, while it is well known that the outcome of a measured qubit depends only on operations in its causal light cone \cite{Abedi2023quantumlazytraining}, there is no systematic approach that exploits this insight for hybrid program optimization when non-contributory qubits are obscured by complex classical post-processing.

\begin{figure}[htb]
\definecolor{qcfill}{RGB}{221,229,248}
\definecolor{qcstroke}{RGB}{40,95,220}
\definecolor{hostfill}{RGB}{247,224,230}
\definecolor{hoststroke}{RGB}{227,120,148}
\definecolor{analysisfill}{RGB}{231,244,234}
\definecolor{analysisstroke}{RGB}{116,193,137}
\definecolor{pillgreen}{RGB}{182,231,194}
\definecolor{pillyellow}{RGB}{246,214,124}
\definecolor{pillorange}{RGB}{244,181,89}
\definecolor{lavfill}{RGB}{239,232,249}
\definecolor{lavstroke}{RGB}{145,98,226}
\definecolor{textgreen}{RGB}{20,120,55}
\definecolor{softgray}{RGB}{242,242,242}

\tikzset{
  >={Latex[length=2.6mm,width=1.8mm]},
  every node/.style={font=\sffamily},
  mainbox/.style={rounded corners=8pt, line width=0.9pt},
  pill/.style={rounded corners=10pt, draw=none, font=\bfseries\small, inner xsep=7pt, inner ysep=4pt},
  proc/.style={rounded corners=4pt, draw=green!45!black, fill=green!10, minimum width=1.8cm, minimum height=0.66cm, align=center, font=\bfseries\small},
  procsmall/.style={rounded corners=3pt, draw=green!45!black, fill=green!10, minimum width=0.58cm, minimum height=0.42cm, align=center, font=\scriptsize\bfseries},
  flow/.style={->, line width=1.0pt},
  note/.style={font=\scriptsize, align=center},
  tinylabel/.style={font=\scriptsize\bfseries},
}

\newcommand{\cfgicon}[2][1.0]{%
\begin{scope}[shift={#2}, scale=#1, line width=0.55pt]
  \node[
    draw=lavstroke,
    fill=lavfill,
    rounded corners=2pt,
    minimum width=0.34cm,
    minimum height=0.16cm
  ] (t) at (0,0.56) {};

  \node[
    draw=qcstroke!80!black,
    fill=white,
    rounded corners=2pt,
    minimum width=0.26cm,
    minimum height=0.14cm
  ] (m) at (0,0.22) {};

  \node[
    diamond,
    draw=green!55!black,
    fill=green!10,
    aspect=1.45,
    minimum size=0.30cm
  ] (d) at (0,-0.16) {};

  \node[
    draw=qcstroke!80!black,
    fill=white,
    rounded corners=2pt,
    minimum width=0.24cm,
    minimum height=0.13cm
  ] (l) at (-0.30,-0.58) {};

  \node[
    draw=qcstroke!80!black,
    fill=white,
    rounded corners=2pt,
    minimum width=0.24cm,
    minimum height=0.13cm
  ] (r) at (0.30,-0.58) {};

  \node[
    draw=qcstroke!80!black,
    fill=white,
    circle,
    minimum size=0.12cm,
    inner sep=0pt
  ] (b) at (0,-0.96) {};

  \draw (t)--(m) (m)--(d);
  \draw (d)--(l) (d)--(r);
  \draw (l)--(b) (r)--(b);
\end{scope}}

\newcommand{\chipicon}[2][1.0]{%
\begin{scope}[shift={#2}, scale=#1, line width=0.6pt]
  \draw[rounded corners=2pt] (-0.18,-0.18) rectangle (0.18,0.18);
  \foreach \x in {-0.12,0,0.12} {
    \draw (\x,0.18)--(\x,0.28);
    \draw (\x,-0.18)--(\x,-0.28);
    \draw (0.18,\x)--(0.28,\x);
    \draw (-0.18,\x)--(-0.28,\x);
  }
  \draw[rounded corners=1pt] (-0.08,-0.08) rectangle (0.08,0.08);
  \draw (-0.03,0.03)--(0.05,0.05);
  \draw (-0.03,0.03)--(0.03,-0.03);
\end{scope}}

\newcommand{\cogicon}[2][1.0]{%
\begin{scope}[shift={#2}, scale=#1, line width=0.7pt, draw=gray!75!black, fill=gray!20]
  \foreach \a in {0,45,...,315} {
    \begin{scope}[rotate=\a]
      \draw[fill=gray!20] (0.22,-0.04) rectangle (0.32,0.04);
    \end{scope}
  }
  \draw[fill=gray!25] (0,0) circle (0.22);
  \draw[fill=white] (0,0) circle (0.08);
\end{scope}}

\newcommand{\measicon}[2][1.0]{%
\begin{scope}[shift={#2}, scale=#1, line width=0.55pt]
  \draw[rounded corners=1pt, fill=white] (-0.18,-0.18) rectangle (0.18,0.18);
  \draw (-0.08,0.02) arc[start angle=180,end angle=0,radius=0.08];
  \draw (0,0.02)--(0.08,0.10);
  \fill (0,0.02) circle (0.012);
\end{scope}}

\newcommand{\analysisicon}[2][1.0]{%
\begin{scope}[shift={#2}, scale=#1, line width=0.6pt]
  \draw[black] (0,0) circle (0.16);
  \draw[black] (0.11,-0.11) -- (0.24,-0.24);

  \node[
    draw=green!55!black, fill=green!10,
    diamond, aspect=1.2, minimum size=0.11cm, inner sep=0pt
  ] at (-0.05,0.03) {};
  \node[
    draw=qcstroke!80!black, fill=white,
    rounded corners=1.2pt,
    minimum width=0.10cm, minimum height=0.06cm, inner sep=0pt
  ] at (0.07,0.06) {};
  \node[
    draw=qcstroke!80!black, fill=white,
    circle, minimum size=0.07cm, inner sep=0pt
  ] at (0.03,-0.06) {};

  \draw[-{Latex[length=1.2mm]}, black] (-0.01,0.01) -- (0.04,0.04);
  \draw[-{Latex[length=1.2mm]}, black] (-0.02,-0.01) -- (0.01,-0.05);
\end{scope}}

\newcommand{\cpuicon}[2][1.0]{%
\begin{scope}[shift={#2}, scale=#1, line width=0.6pt]
  \draw[black, rounded corners=2pt, fill=white] (-0.22,-0.22) rectangle (0.22,0.22);

  \foreach \y in {-0.12,0,0.12}{
    \draw[black] (-0.28,\y) -- (-0.22,\y);
    \draw[black] (0.22,\y) -- (0.28,\y);
  }
  \foreach \x in {-0.12,0,0.12}{
    \draw[black] (\x,0.22) -- (\x,0.28);
    \draw[black] (\x,-0.22) -- (\x,-0.28);
  }

  \draw[black, fill=qcstroke!10] (-0.10,-0.10) rectangle (0.10,0.10);

  \node[font=\tiny\bfseries] at (-0.6,0) {CPU};
\end{scope}}

\newcommand{\gpuicon}[2][1.0]{%
\begin{scope}[shift={#2}, scale=#1, line width=0.6pt]
  \draw[black, rounded corners=2pt, fill=white] (-0.28,-0.22) rectangle (0.28,0.22);

  \foreach \y in {-0.12,0,0.12}{
    \draw[black] (-0.36,\y) -- (-0.28,\y);
    \draw[black] (0.28,\y) -- (0.36,\y);
  }
  \foreach \x in {-0.18,-0.06,0.06,0.18}{
    \draw[black] (\x,0.22) -- (\x,0.30);
    \draw[black] (\x,-0.22) -- (\x,-0.30);
  }

  \foreach \x in {-0.12,0,0.12}{
    \foreach \y in {-0.07,0.07}{
      \draw[black, fill=green!10] (\x-0.04,\y-0.04) rectangle (\x+0.04,\y+0.04);
    }
  }

  \node[font=\tiny\bfseries] at (-0.6,0) {GPU};
\end{scope}}

\newcommand{\ssaicon}[2][1.0]{%
\begin{scope}[shift={#2}, scale=#1, line width=0.55pt]

  \draw[
    dashed, dash pattern=on 3pt off 2pt,
    draw=lavstroke, fill=lavfill!45,
    rounded corners=5pt
  ] (-0.95,-0.78) rectangle (0.95,0.72);

  \node[font=\scriptsize\bfseries, text=lavstroke] at (0,0.6) {levelized SSA};

  \draw[dashed, dash pattern=on 2pt off 1.5pt, draw=lavstroke!80]
    (-0.78,0.18) -- (0.78,0.18);
  \draw[dashed, dash pattern=on 2pt off 1.5pt, draw=lavstroke!80]
    (-0.78,-0.18) -- (0.78,-0.18);

  \node[font=\tiny\bfseries, text=lavstroke] at (-0.68,0.35) {$L_0$};
  \node[font=\tiny\bfseries, text=lavstroke] at (-0.68,0.00) {$L_1$};
  \node[font=\tiny\bfseries, text=lavstroke] at (-0.68,-0.35) {$L_2$};

  \node[
    draw=green!55!black, fill=green!10, rounded corners=2pt,
    minimum width=0.28cm, minimum height=0.17cm, inner sep=1pt
  ] at (-0.10,0.35) {$s_1$};

  \node[
    draw=green!55!black, fill=green!10, rounded corners=2pt,
    minimum width=0.28cm, minimum height=0.17cm, inner sep=1pt
  ] at (0.28,0.35) {$s_2$};

  \node[
    draw=green!55!black, fill=green!10, rounded corners=2pt,
    minimum width=0.28cm, minimum height=0.17cm, inner sep=1pt
  ] at (0.62,0.35) {$s_3$};

  \node[
    draw=green!55!black, fill=green!10, rounded corners=2pt,
    minimum width=0.28cm, minimum height=0.17cm, inner sep=1pt
  ] at (0.05,0.00) {$s_4$};

  \node[
    draw=green!55!black, fill=green!10, rounded corners=2pt,
    minimum width=0.28cm, minimum height=0.17cm, inner sep=1pt
  ] at (0.43,0.00) {$s_5$};

  \node[
    draw=green!55!black, fill=green!10, rounded corners=2pt,
    minimum width=0.28cm, minimum height=0.17cm, inner sep=1pt
  ] at (0.22,-0.35) {$s_6$};

  \node[font=\tiny\bfseries, text=lavstroke] at (0,-0.62) {parallel levels};

\end{scope}}

\newcommand{\qcicon}[2][1.0]{%
\begin{scope}[shift={#2}, scale=#1, line width=0.75pt]
  \foreach[count=\i from 0] \y in {0,-0.36,-0.72,-1.08} {
    \draw (0,\y) -- (2.80,\y);
    \node[left, font=\scriptsize\itshape] at (-0.10,\y) {$q_{\i}$};
  }

  \draw[fill=white] (0.36,-0.89) rectangle (0.76,0.23);
  \draw[fill=white] (1.10,-0.17) rectangle (1.46,0.18);
  \draw[fill=white] (1.73,-0.91) rectangle (2.18,-0.12);
  \draw[fill=white] (1.10,-1.23) rectangle (1.46,-0.88);

  \foreach \y in {0,-0.36,-0.72,-1.08} {
    \measicon[0.78]{(2.81,\y)}
  }
\end{scope}}

\newcommand{\redxmark}[4][0.12]{%
\draw[red, line width=#4] (#2-#1,#3-#1) -- (#2+#1,#3+#1);
\draw[red, line width=#4] (#2-#1,#3+#1) -- (#2+#1,#3-#1);
}

\newcommand{\qcopticon}[2][1.0]{%
\begin{scope}[shift={#2}, scale=#1, line width=0.75pt]
  \foreach[count=\i from 0] \y in {0,-0.36,-0.72,-1.08} {
    \draw (0,\y) -- (2.80,\y);
    \node[left, font=\scriptsize\itshape] at (0.1,\y) {$q_{\i}$};
  }

  \draw[fill=white] (0.36,-0.89) rectangle (0.76,0.23);

  \draw[fill=white] (1.10,-0.17) rectangle (1.46,0.18);
  \redxmark[0.13]{1.28}{0.005}{0.85pt}

  \draw[fill=white] (1.73,-0.91) rectangle (2.18,-0.12);

  \draw[fill=white] (1.10,-1.23) rectangle (1.46,-0.88);
  \redxmark[0.13]{1.28}{-1.04}{0.85pt}

  \foreach \y in {0,-0.36,-0.72,-1.08} {
    \measicon[0.78]{(2.81,\y)}
  }

  \redxmark[0.14]{2.81}{0.00}{0.9pt}
  \redxmark[0.14]{2.81}{-1.08}{0.9pt}
\end{scope}}

\newcommand{\hybridicon}[2][1.0]{%
\begin{scope}[shift={#2}, scale=#1, transform shape]
  \draw[mainbox, draw=gray!70, fill=softgray] (0,0.2) rectangle (3.7,4.55);
  \node[pill, fill=gray!25!white] at (1.93,4.60) {hybrid program};

  \draw[mainbox, draw=qcstroke, fill=qcfill] (0.22,2.30) rectangle (3.63,4.02);
  \node[pill, fill=qcstroke!25!white] at (1.38,4.07) {quantum circuit};

  \qcicon[0.92]{(0.67,3.56)};

  \draw[
    rounded corners=3pt,
    dashed,
    draw=lavstroke,
    line width=0.7pt
  ] (3.02,2.36) rectangle (3.5,3.70);

  \draw[mainbox, draw=hoststroke, fill=hostfill] (0.22,0.30) rectangle (3.63,1.95);
  \node[pill, fill=hoststroke!20!white] at (1.08,2.00) {host CFG $H$};

  \cfgicon[0.72]{(1.82,1.22)}

  \coordinate (hostinput) at (1.96,1.60);

  \draw[
    -{Latex[length=2.0mm,width=1.6mm]},
    draw=lavstroke,
    line width=0.9pt
  ]
  (3.02,2.36)
  .. controls (3.02,1.96) and (2.46,1.88) ..
  (hostinput);

  \node[
    font=\scriptsize\bfseries,
    text=lavstroke,
    align=center
  ] at (3.05,1.48)
  {measurement\\outcomes\\$o_0,\dots,o_3$};
\end{scope}}

\newcommand{\hostanalysisicon}[2][1.0]{%
\begin{scope}[shift={#2}, scale=#1, transform shape]

  \draw[mainbox, draw=analysisstroke, fill=analysisfill] (0,0.3) rectangle (5.15,4.35);
  \node[pill, fill=pillgreen] at (2.58,4.30) {host-side static analysis};

  \draw[
    rounded corners=7pt,
    draw=qcstroke!35,
    fill=blue!6
  ] (0.18,2.45) rectangle (4.97,3.86);

  \node[font=\scriptsize\bfseries, text=qcstroke] at (4.22,3.63) {sequential analysis};

  \cfgicon[0.68]{(0.58,3.25)};
  \draw[flow] (0.90,3.25) -- (3.52,3.25);

  \analysisicon[0.92]{(3.76,3.25)};
  \node[font=\scriptsize\bfseries] at (4.42,3.25) {analysis};

  \cpuicon[0.82]{(4.58,2.73)};

  \draw[
    rounded corners=7pt,
    draw=green!45!black,
    fill=green!6
  ] (0.18,0.40) rectangle (4.97,2.28);

  \node[font=\scriptsize\bfseries, text=green!45!black] at (4.16,2.05) {parallel analysis};

  \cfgicon[0.68]{(0.58,1.36)};

  \draw[flow] (0.90,1.36) -- (1.6,1.36);
  \node[font=\scriptsize\bfseries] at (1.25,1.58) {lower};

  \ssaicon[0.80]{(2.41,1.36)};

  \draw[flow] (3.15,1.36) -- (3.55,1.36);

  \analysisicon[0.92]{(3.76,1.36)};
  \node[font=\scriptsize\bfseries] at (4.42,1.36) {analysis};

  \gpuicon[0.80]{(4.58,0.86)};

\end{scope}}

\newcommand{\optimizericon}[2][1.0]{%
\begin{scope}[shift={#2}, scale=#1, transform shape]
  \draw[
    mainbox,
    draw=orange!80!black,
    fill={rgb,255:red,250;green,234;blue,213}
  ] (0,0.5) rectangle (2.35,1.60);

  \node[pill, fill=pillorange] at (1.05,1.58) {optimizer};

  \cogicon[0.82]{(0.52,0.85)}

  \node[
    align=left,
    font=\scriptsize\bfseries
  ] at (1.68,0.80) {DGE-based\\pruning};
\end{scope}}

\newcommand{\resulticon}[2][1.0]{%
\begin{scope}[shift={#2}, scale=#1, transform shape]
  \draw[mainbox, draw=analysisstroke, fill=white] (0.5,0.8) rectangle (2.40,1.60);

  \node[pill, fill=pillgreen] at (1.2,1.62) {result of analysis};

  \node[
    font=\scriptsize\bfseries,
    text=textgreen,
    align=center
  ] at (1.45,0.98) {$o_0,o_3$ are not\\ contributory};
\end{scope}}

\begin{tikzpicture}[x=1cm,y=1cm]

\hybridicon[1.0]{(-3.00,4.25)}

\hostanalysisicon[1.0]{(3.0,4.35)}

\optimizericon[1.0]{(-2.5,1.15)}

\draw[mainbox, draw=qcstroke, fill=qcfill] (4.55,1.55) rectangle (8.,3.10);
\node[pill, fill=qcstroke!25!white] at (6.55,3.18) {optimized circuit};
\qcopticon[0.9]{(5.0,2.65)}

\draw[
  flow,
  draw=hoststroke
]
(0.8,6.52) -- (3,6.52);

\draw[
  flow,
  draw=qcstroke
]
(-1,4.43)
.. controls (-1,3.85) and (-1,3.20) ..
(-1,2.95);

\qcicon[0.5]{(-2.2, 4)}

\draw[
  -{Latex[length=2.4mm,width=1.8mm]},
  draw=green!55!black,
  line width=1.0pt
]
(4.10,4.62)
.. controls (3.80,3.70) and (2.35,2.95) ..
(-0.2,2.5);

\node[
  fill=green!6,
  draw=green!45!black,
  rounded corners=2pt,
  align=center,
  inner xsep=4pt,
  inner ysep=2pt,
  font=\scriptsize\bfseries,
  text=green!40!black
] at (3.45,3.85)
{analytic result:\\
$o_0$ and $o_3$ are non-contributory};

\draw[
  flow,
  draw=orange!85!black
]
(-0.17,2.25) .. controls (1, 2.10) .. (4.52,2.00);

\end{tikzpicture}

    \centering
    \caption{Overview of our optimization pass. Given a hybrid quantum program, we analyze the classical host program to determine which measurement outcomes are non-contributory to the final return. These analysis results are then passed to the optimizer, which performs DGE-based pruning on the circuit and removes measurements and gates that affect only non-contributory outcomes. The host-side analysis can be executed either sequentially on a CPU backend or, after lowering to a levelized SSA representation, by a CUDA-based GPU backend.
}
    \label{fig:framework_overview}
\end{figure}
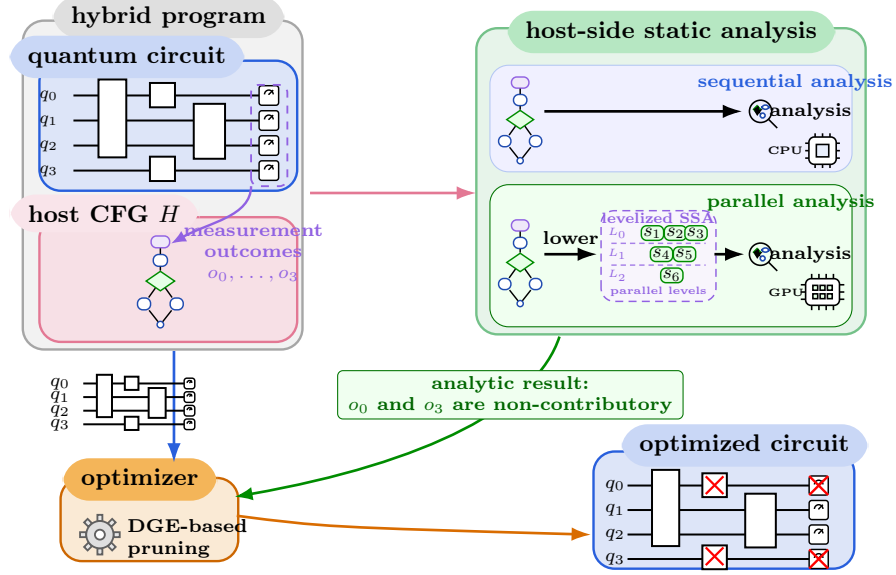

\noindent\textbf{Our contribution i.}
We present a semantics-aware host-side static analysis
for hybrid quantum programs that identifies measurement outcomes that are
syntactically read by the host but semantically non-contributory to the host
return. The analysis is formulated as an abstract interpretation over the
classical host program, tracking data and control contribution of
measurement-bound initial values with exact symbolic reasoning when possible
and conservative dependency information otherwise. We formalize semantic
contribution and prove soundness: every reported non-contributory outcome is
irrelevant to the host program's returned value. We evaluate the analysis on
$24$ application-faithful hybrid workloads across four domains, with host computations derived from upstream application artifacts, and include a standard syntactic liveness analysis as baseline. We show that the result of the host-side analysis enables around $38\%$ mean total-gate reduction standalone, and still removes about $31\%$ of the post-optimized gates after SOTA circuit optimizers Qiskit \cite{Qiskit}, t|ket$\rangle$ \cite{Sivarajah_tket_2021}, and PyZX \cite{kissinger2019pyzx} have already run.
Besides, our method detects $4\times$ as many non-contributory measurements as a baseline standard syntactic liveness analysis.
This indicates that the host-semantic simplifications exposed by our analysis are not subsumed by by either widely used circuit-local optimizers or standard syntactic liveness analysis.

\noindent\textbf{A scalability challenge.}
The host-side static analysis introduced by our method is the main semantic phase that goes beyond circuit-only optimisation and, therefore, a primary source of compile-time cost in our flow.
Accelerating this phase is appealing, but it is not straightforward: in its original form, the analysis exposes irregular dependences and limited explicit parallel structure, making it a poor match for CUDA's SIMT execution model~\cite{cuda_programming_guide,cuda_best_practices}.
This motivates an SSA-style reformulation. SSA makes definitions and data dependences explicit~\cite{cytron1991ssa}, and prior work has studied both the connection between SSA translation and abstract interpretation, and the semantics preservation of SSA-based compiler middle ends~\cite{lemerre2023ssa,barthe2014ssa,10.1145/2499370.2462164}. Inspired by this line of work, we use SSA in a narrower setting: to expose level-wise independent abstract-transfer tasks while preserving the host-side analysis result.
Indeed, as shown empirically in \hyperref[par:why-lowering-matters]{the last report} in \cref{subsec:gpu-acc}, directly moving the source-order analysis to the GPU is ineffective.
This observation motivates a dedicated reformulation of the host-side analysis for GPU execution, which we describe next.

\noindent\textbf{Our contribution ii.}
We introduce an SSA-style, levelized intermediate representation for the host program to expose sufficient parallel structure for effective GPU acceleration of the host-side analysis. We prove that this lowering
preserves the set of non-contributory outcomes computed by the host-side analysis.
On top of this lowered representation, we implement a CUDA-based backend and evaluate the resulting GPU-accelerated host-side analysis, showing that the lowered backend achieves substantial speedups over a sequential baseline.

\noindent\textbf{Overview of optimization flow.} \Cref{fig:framework_overview} provides an overview of the full flow of using our analysis to guide circuit optimization.


\section{Method}\label{sec:methods}
\paragraph{Scope and assumptions} In this manuscript, we restrict our consideration to the following host program fragment:
\[
\begin{array}{rcl}
S &::=& \texttt{skip}
      \mid \texttt{x} := e
      \mid \texttt{return}\ e
      \mid S;S
      \mid \texttt{if}\ (e)\ \texttt{then}\ S\ \texttt{else}\ S,\\[1mm]
e &::=& \texttt x
      \mid c
      \mid -e
      \mid e_1 \odot e_2
      \mid \texttt{int}(e)
      \mid \texttt{random}(L,U)
      \mid \texttt{f}(e_1,\ldots,e_n),
\end{array}
\]
where \(\texttt x\) is a program variable, \(\odot \in \{+,-,\times,\div,\%\}\), and rational
constants are interpreted exactly. 
Throughout the paper, \(\mathbb{Q}\) denotes the set of rational numbers.
The operation \(\texttt{int}(\cdot)\) denotes truncation toward zero. \(\texttt{random}(L,U)\) has measurement-independent constant bounds. Function calls are assumed
to be total, pure, and side-effect free.
The grammar supports runtime-dependent structured conditionals, including nested conditionals. 
Arbitrary loops and recursion are outside the current implementation scope. Standard CFG-based fixpoint iteration with joins/widening, optionally combined with bounded unrolling, can extend the analysis to loops.

Next, we design abstract interpretation with a custom domain to symbolically evaluate the program.

\subsection{Abstract domain}\label{subsec:abs-dom}

$\mathcal{V}$ denotes the set of all program variables. $\mathbf{M}\subseteq\mathcal{V}$ is a set of variables that take measurement outcome as initial value.
Our goal is to perform a forward analysis on the host program to track which initial values of variables in $\mathbf{M}$ influence the final return.
$\Sigma_0=\{\texttt{x}@0 \mid \texttt{x}\in \mathcal{V}\}$ denotes the set of symbols representing initial values: for each variable $\texttt{x}$, $\texttt{x}@0$ denotes its initial value.

\paragraph{Program state}
An abstract state of the program is a mapping $\rho : \mathcal{V}\cup\{\texttt{return}\} \rightarrow \mathcal{A}
$ that maps each program variable to an abstract value.
We reserve a distinguished symbol \(\texttt{return}\notin \mathbf{M}\) to hold the abstract return value, i.e., we use $\rho(\texttt{return})$ to denote the abstract return value after executing the current statement. With $\rho(\texttt{return}) = \bot$, it means current statement has no return.
Besides, we maintain two auxiliary maps $\mathbf{Dom}: \Sigma \rightarrow \{\textbf{bin}, \textbf{int}, \textbf{real}\}$, and $\mathbf{B}: \Sigma \rightarrow \{\bigcup_k[L_k, U_k] \mid L, U \in \mathbb{Q}\}$, where $\Sigma$ is the set of all symbols, $\mathbf{Dom}$ records the type/domain of each symbol,
$\mathbf{B}$ records the value range for every symbol.
We maintain \(C_{\texttt{ctrl}} \subseteq \mathbf{M}\) as the set of
measurement-bound variables whose initial values affect the evaluation of at least one conditional branch in a way that changes the observable post-conditional result, in particular the eventual abstract return.

\paragraph{Abstract value}
An abstract value \(\mathcal A\) is either a polynomial form \(\langle p\rangle\)
or a dependence form \(\langle D\rangle\):
\[
\mathcal A \mathrel{::=} \langle p\rangle \mid \langle D\rangle,
\]
where \(p\in\mathcal P\) is a multivariate polynomial over symbols in
\(\Sigma\) with rational coefficients, and \(D\subseteq\Sigma\) is a set of
symbols the abstract value may depend on.

In our analysis, an abstract value is either an exact polynomial, if we can precisely compute it within our domain (see \cref{ex:abstract-value-poly-form}) or a set of symbols it may depend on, if exact arithmetic is not feasible (see \cref{ex:abstract-value-dependency-form}). 
The dependence form is a conservative fallback recording initial inputs that might influence the value.
The polynomial form preserves more information: 
when $\mathcal{A} = \langle p\rangle$, $\mathcal{A}$ exactly depends on all symbols contained in the polynomial $p$.

\begin{example}\label{ex:abstract-value-poly-form}
A program with one statement:
$\texttt{x} := 2\texttt{a} - 3\texttt{b} + 5$. After its execution, the abstract state, if expressed in polynomial form for all variables, is $\rho(\texttt{a})=\big\langle \texttt{a}@0\big\rangle$, $\rho(\texttt{b})=\big\langle\texttt{b}@0\big\rangle$, $\rho(\texttt{x})=\big\langle 2\texttt{a}@0 - 3\texttt{b}@0 + 5\big\rangle$; if expressed in dependence form for all variables, is $\rho(\texttt{a})=\big\langle \{\texttt{a}@0\}\big\rangle$, $\rho(\texttt{b})=\big\langle\{\texttt{b}@0\}\big\rangle$ $\rho(\texttt{x})=\big\langle\{\texttt{a}@0, \texttt{b}@0\}\big\rangle$. We use polynomial forms whenever possible to hold more information.

\end{example}

\begin{example}\label{ex:abstract-value-dependency-form}
A program with one statement:
$\texttt{x} := 4\texttt{f}(\texttt{b}, \texttt{a}) - 1$, where $\texttt{f}$ is a function call with unknown implementation. After its execution, 
$\rho(\texttt{x})=\big\langle\{\texttt{a}@0, \texttt{b}@0\}\big\rangle$, since we lose track of the exact symbolic evaluation on $\texttt{x}$, a polynomial form for $\rho(\texttt{x})$ is impossible. $\rho(\texttt{x})=\big\langle\{\texttt{a}@0, \texttt{b}@0\}\big\rangle$ only preserves the information that $\texttt{x}$ depends on.

\end{example}

\paragraph{Exact abstract equality}
Two abstract values are treated as equal only when they are both exact polynomial forms with the same normalized polynomial:
\[
\mathsf{Eq}_{\mathrm A}(A,B)
=
\begin{cases}
\texttt{true}, &
A=\langle p\rangle,\; B=\langle q\rangle,\; \mathrm{norm}(p-q)=0,\\[1mm]
\texttt{true}, &
A=B=\bot,\\[1mm]
\texttt{false}, &
\text{otherwise,}
\end{cases}
\]
where $\mathrm{norm}$ is a canonical polynomial simplification which collects like terms, orders monomials canonically, and drops zero terms. For example,
$\mathrm{norm}((\texttt{a}@0+\texttt{b}@0)-(\texttt{b}@0+\texttt{a}@0))=0$ and
$\mathrm{norm}(2\,\texttt{a}@0\cdot\texttt{b}@0+3\,\texttt{b}@0\cdot\texttt{a}@0)=5\,\texttt{a}@0\cdot\texttt{b}@0$.

\subsection{Abstract evaluation}\label{subsec:abs-eval}
We define abstract evaluation of expressions. $[\![e]\!]^{\sharp}$ denotes the abstract evaluation of an expression $e$ under environment $\rho$. In the following, we use the notation $\mathbf{Sym}(\cdot)$ very often: given any abstract value $\mathcal{A}$, $\mathbf{Sym}(\mathcal{A})$ is the set of all symbols occurring in $\mathcal{A}$, including symbols in $\Sigma_0$ and any fresh symbols introduced by \texttt{random}(\(\cdot,\cdot\)), for example,
$\mathbf{Sym}\!\left(\,\big\langle 2\,\texttt{a}@0 \;+\; 4\texttt{b}@0 \;-\; 3 \big\rangle\,\right)
\;=\; \{\, \texttt{a}@0,\ \texttt{b}@0 \,\}$,
$\mathbf{Sym}\!\left(\,\big\langle \{\texttt{a}@0, \texttt{b}@0, \texttt{c}@0\} \big\rangle\,\right)
\;=\; \{\, \texttt{a}@0,\ \texttt{b}@0,\ \texttt{c}@0 \,\}$.
We set $\mathbf{Sym}(\bot)$ to $\varnothing$.

\paragraph{Base case} 
\[
\frac{\texttt{x}\in\mathcal{V}}{[\![\texttt{x}]\!]^{\sharp}=\rho (\texttt{x})}
\quad
\frac{c\in\mathbb{Q}}{[\![c]\!]^{\sharp}=\langle c\rangle}
\quad
\frac{c\notin\mathbb{Q}}{[\![c]\!]^{\sharp}=\langle \varnothing\rangle}
\]

\paragraph{Binary and unary arithmetic}
When both operands have exact polynomial form, we compute the result exactly in the polynomial ring, using $+$, $-$, $\times$, and then apply canonical polynomial simplification $\mathrm{norm}$:
\[\frac{[\![e_1]\!]^{\sharp}=\langle p\rangle \wedge [\![e_2]\!]^{\sharp}=\langle q\rangle\wedge p, q\in \mathcal{P}\wedge \odot\in\{+,-,\times\}}{[\![e_1 \odot e_2]\!]^{\sharp}=\langle \mathrm{norm}(p\ \odot\ q)\rangle}\]
To keep the analysis efficient, we use a simplified polynomial domain: since division and modulo are not closed, we drop the polynomial form for these operators and return a dependence form:
\[
\frac{\mathcal{A}_1=[\![e_1]\!]^{\sharp}\wedge \mathcal{A}_2=[\![e_2]\!]^{\sharp}\wedge \odot\in\{\div, \%\}}{[\![e_1 \odot e_2]\!]^{\sharp}=\langle \mathbf{Sym}(\mathcal{A}_1)\cup \mathbf{Sym}(\mathcal{A}_2)\rangle}
\]
Similarly, if either operand is already in dependence form $\langle D\rangle$:
\[
\frac{\mathcal{A}_1=[\![e_1]\!]^{\sharp}=\langle D\rangle\wedge  D \subseteq \Sigma\wedge \mathcal{A}_2=[\![e_2]\!]^{\sharp} }{[\![e_1 \odot e_2]\!]^{\sharp}=[\![e_2 \odot e_1]\!]^{\sharp}=\langle \mathbf{Sym}(\mathcal{A}_1)\cup \mathbf{Sym}(\mathcal{A}_2)\rangle}
\]
Unary minus is value-preserving on dependencies: 
\[
\frac{[\![e]\!]^{\sharp}=\langle p\rangle\wedge p\in\mathcal{P}}{[\![-e]\!]^{\sharp}=\langle -p\rangle}
\qquad
\frac{[\![e]\!]^{\sharp}=\langle D\rangle\wedge D\subseteq\Sigma}{[\![-e]\!]^{\sharp}=\langle D\rangle}
\]

\paragraph{Operator \texttt{int}$(\cdot)$}
In this manuscript, \texttt{int}$(\cdot)$ is truncating toward zero: it discards the fractional part, and is not a flooring (e.g., $\texttt{int}(1.2)=1$, $\texttt{int}(-1.5)=-1$, $\texttt{int}(-0.3)=0$).
If $[\![e]\!]^{\sharp}$ evaluates to an exact polynomial that is integer-valued under the domain map $\mathbf{Dom}$, then $[\![\texttt{int}(e)]\!]^{\sharp}$ returns the same polynomial:
\[
\frac{C_{int}^1(e) \;\equiv\; [\![e]\!]^{\sharp}=\langle p\rangle \wedge \mathrm{IsInt}(p,\mathbf{Dom}) = \texttt{true}}
     {[\![\texttt{int}(e)]\!]^{\sharp}=\langle p\rangle}
\]
where $\mathrm{IsInt}(p,\mathbf{Dom})$ holds iff all coefficients of $p$ are integers and
$p$ mentions no symbol whose domain is \textbf{real}.

If $[\![e]\!]^{\sharp}$ evaluates to an exact polynomial
$p$ whose evaluation gives a value range $I$ under the bound map $\mathbf{B}$, and this value range lies entirely within a single truncation bin, then $[\![\texttt{int}(e)]\!]^{\sharp}$ constant-folds to an integer:
\[
\frac{C_{int}^2(e) \;\equiv\; [\![e]\!]^{\sharp}=\langle p\rangle \wedge \mathrm{eval}(p,\mathbf{B})=I \wedge \mathrm{trunc}(I)=c}
     {[\![\texttt{int}(e)]\!]^{\sharp}=\langle c\rangle}
\]
where the operator
$\mathrm{eval}(p,\mathbf{B})$ is the value range of $p$ evaluated under the bounds of variables given by $\mathbf{B}$;
$\forall i\in \mathbb{Z}: \mathrm{trunc}(S)=i$ iff $\forall v\in S:\texttt{int}(v)=i$.

We further use the following integer-carrier rule: if an exact polynomial can be decomposed into an integer-valued part plus a residual whose feasible range cannot push the value across any truncation boundary, then truncation preserves the integer-valued part:
\[
\frac{
\begin{array}{l}
C_{int}^3(e) \;\equiv\; [\![e]\!]^{\sharp}=\langle r\rangle \wedge
\exists p,q\in\mathcal{P}:\; r=\mathrm{norm}(p+q)\wedge\mathrm{IsInt}(p,\mathbf{Dom})=\texttt{true} \wedge\\
\mathrm{eval}(p,\mathbf{B})=I_p \wedge
\mathrm{eval}(q,\mathbf{B})=I_q \wedge
\mathrm{NoCarry}(I_p,I_q)=\texttt{true}
\end{array}
}
{[\![\texttt{int}(e)]\!]^{\sharp}=\langle p\rangle}
\]
where $\mathrm{NoCarry}(I_p,I_q)$ holds iff
$\forall k\in (I_p\cap\mathbb{Z}),\ \forall \delta\in I_q:\ \texttt{int}(k+\delta)=k$.

A common sufficient pattern is \(I_q\subseteq [0,1)\) with \(I_p\subseteq [0,\infty)\), or symmetrically \(I_q\subseteq (-1,0]\) with \(I_p\subseteq (-\infty,0]\).
In particular, if \(\eta\in(0.1,0.5)\) and \(\texttt{a},\texttt{b},\texttt{c}\in\{0,1\}\), then for
\(p=\texttt{b}\cdot\texttt{c}\) and \(q=\eta\cdot\texttt{a}\), we have \(I_p\cap\mathbb{Z}=\{0,1\}\) and \(I_q\subseteq [0,0.5]\), so \(\mathrm{NoCarry}(I_p,I_q)\) holds and
\[
[\![\texttt{int}(\eta\cdot\texttt{a}+\texttt{b}\cdot\texttt{c})]\!]^{\sharp}
=
\langle \texttt{b}\cdot\texttt{c}\rangle.
\]
If none of the three cases holds, then:
\[
\frac{\neg C_{int}^1(e) \wedge \neg C_{int}^2(e) \wedge \neg C_{int}^3(e)}
     {[\![\texttt{int}(e)]\!]^{\sharp}=\langle \mathbf{Sym}([\![e]\!]^{\sharp})\rangle }
\]

\paragraph{Operator \texttt{random}$(\cdot, \cdot)$}
Each occurrence of \(\texttt{random}(L,U)\) is abstracted as a fresh real-valued
symbol \(\texttt r\) with bounds \([L,U]\). This rule is used only when
\(L,U\in\mathbb Q\) are constants independent of all measurement-bound inputs.
The generated symbol does not coincide with any symbol already used in the
current analysis and carries no measurement dependency:
\[
\frac{
\begin{array}{c}
L, U\in\mathbb Q
\wedge
\texttt{r}=\mathrm{gen}\text{-}\mathrm{fresh}()\wedge
\mathbf{Dom}(\texttt{r})=\textbf{real}
\wedge
\mathbf{B}(\texttt{r})=[L,U]
\end{array}
}
{[\![\texttt{random}(L,U)]\!]^{\sharp}=\langle \texttt{r}\rangle}
\]
where \(\mathrm{gen}\text{-}\mathrm{fresh}()\) picks a fresh symbol for this
particular occurrence of \(\texttt{random}\).

\paragraph{Function calls}
We assume that function calls appearing in expressions are total, pure, and
side-effect free, and the return can depend only on the values of arguments. We conservatively track the union of the argument supports:
\[
\frac{
n\in\mathbb{N}
\wedge
\bigcup_{i=1}^{n}
\mathbf{Sym}([\![\texttt{a}_i]\!]^{\sharp})
=
D_{\texttt{f}}
}
{
[\![ \texttt{f}(\texttt{a}_1,\dots,\texttt{a}_n) ]\!]^{\sharp}
=
\langle D_{\texttt{f}}\rangle
}
\]

\subsection{Abstract effects}\label{subsec:abs-effect}
Abstract effects define how statements change the abstract state.
We write \((\rho,C_{\texttt{ctrl}})\ \xRightarrow{\ S_b\ }\ (\rho',C_{\texttt{ctrl}}')\) for the abstract execution
of a statement block \(S_b\) from state \((\rho,C_{\texttt{ctrl}})\) to \((\rho',C_{\texttt{ctrl}}')\), where \(S_b\) could be an empty statement, a single statement, or sequential composition of multiple statements. For conciseness, we omit the rules defining abstract effects on $\mathbf{Dom}$ and $\mathbf{B}$.

\paragraph{Sequencing}
For statements $S$, $T$, we write \(S;T\) for sequencing:
\[
\frac{
(\rho,C_{\texttt{ctrl}}) \xRightarrow{ S } (\rho_1,C_{\texttt{ctrl}}^1)
\wedge
(\rho_1,C_{\texttt{ctrl}}^1) \xRightarrow{ T }(\rho_2,C_{\texttt{ctrl}}^2)
}{
(\rho,C_{\texttt{ctrl}})\ \xRightarrow{\ S;T\ }\ (\rho_2,C_{\texttt{ctrl}}^2)
}
\]

\paragraph{Assignments and returns}
Evaluate the expression and update $\rho$:
\[
\frac{\rho'=\rho\oplus\{\texttt{x}\mapsto [\![e]\!]^{\sharp}\}}{(\rho, C_{\texttt{ctrl}})
~\xRightarrow{\ \texttt{x := }e } (\rho', C_{\texttt{ctrl}})}\
\frac{
\rho'=\rho\oplus\{\texttt{return}\mapsto  [\![e]\!]^{\sharp}\}
}{
(\rho,C_{\texttt{ctrl}})\ \xRightarrow{\ \texttt{return }e }\ (\rho',C_{\texttt{ctrl}})
}
\]
where $\oplus$ denotes updating the environment with the information in its second/right operand.

\paragraph{Conditionals}
For $S_{\texttt{cond}}\equiv\texttt{if (cond) then }S_t\texttt{ else }S_e$, we simulate both of the branches from the same input state and then merge their outputs variable-wise: if both branches yield the same abstract value for a variable, we keep it; otherwise we conservatively union their dependencies. 
If the abstract return values differ or some variable either differs between branches or appears in dependence form, then any measurement input that can affect \texttt{cond} is added to $C_{\texttt{ctrl}}$:
\[
\frac{
(\rho,C_{\texttt{ctrl}}) \xRightarrow{S_t}(\rho_t,C_{\texttt{ctrl}}^t)\wedge
(\rho,C_{\texttt{ctrl}}) \xRightarrow{ S_e\ } (\rho_e,C_{\texttt{ctrl}}^e)
}{
(\rho,C_{\texttt{ctrl}})\ \xRightarrow{S_{\texttt{cond}} }\ (\rho_t \sqcup_{\mathrm{env}} \rho_e, C_{\texttt{ctrl}}')
}
\]
where 
\[
(\rho_a \sqcup_{\mathrm{env}} \rho_b)(\texttt{x}) =
\begin{cases}
\rho_a(\texttt{x}), &
\mathsf{Eq}_{\mathrm A}(\rho_a(\texttt{x}),\rho_b(\texttt{x})),\\[1mm]
\big\langle \mathbf{Sym}(\rho_a(\texttt{x})) \cup \mathbf{Sym}(\rho_b(\texttt{x})) \big\rangle, &
\text{otherwise,}
\end{cases}
\]
and
\[
C_{\texttt{ctrl}}' =
C_{\texttt{ctrl}}^t \cup C_{\texttt{ctrl}}^e
\cup
\begin{cases}
\{\texttt{m}\in\mathbf{M}\mid \texttt{m}@0\in\mathbf{Sym}([\![\texttt{cond}]\!]^{\sharp})\},
&
\text{Diff}_{\rho_t,\rho_e},\\[1mm]
\varnothing, &
\text{otherwise,}
\end{cases}
\]
where
$\text{Diff}_{\rho_t,\rho_e}
=
\exists \texttt{x}\in\mathcal{V}\cup\{\texttt{return}\}:\;
\neg \mathsf{Eq}_{\mathrm A}(\rho_t(\texttt{x}),\rho_e(\texttt{x}))
$.

\paragraph{Bounded lookahead for a common post-conditional return}
A common pattern in hybrid host code is a structured conditional with
return-free branches followed by one common return:
$S_{\texttt{if-ret}}
\;\equiv\;
\texttt{if (cond) then } S_t \texttt{ else } S_e \texttt{; return } e$,
where \(S_t\) and \(S_e\) contain no \texttt{return}.
We first analyze the two branches from the same input state, and then
evaluate the common return expression once under each branch
environment:
\[
\frac{
(\rho,C_{\texttt{ctrl}}) \xRightarrow{S_t} (\rho_t,C_{\texttt{ctrl}}^t)
\;\wedge\;
(\rho,C_{\texttt{ctrl}}) \xRightarrow{S_e} (\rho_e,C_{\texttt{ctrl}}^e)
\;\wedge\;
R_t = [\![e]\!]^{\sharp}_{\rho_t}
\;\wedge\;
R_e = [\![e]\!]^{\sharp}_{\rho_e}
}{
(\rho,C_{\texttt{ctrl}})
\;\xRightarrow{S_{\texttt{if-ret}}}\;
\bigl(\rho', C_{\texttt{ctrl}}'\bigr)
}
\]
where
\[
\rho' =
(\rho_t \sqcup_{\mathrm{env}} \rho_e)
\oplus
\{\texttt{return} \mapsto (R_t \sqcup_{\mathrm{A}} R_e)\},
\]
\[
R_t \sqcup_{\mathrm{A}} R_e
=
\begin{cases}
R_t, &
\mathsf{Eq}_{\mathrm A}(R_t,R_e),\\[1mm]
\bigl\langle \mathbf{Sym}(R_t)\cup \mathbf{Sym}(R_e)\bigr\rangle, &
\text{otherwise,}
\end{cases}
\]
and
\[
C_{\texttt{ctrl}}'
=
C_{\texttt{ctrl}}^t \cup C_{\texttt{ctrl}}^e
\cup
\begin{cases}
\{\texttt{m}\in \mathbf{M}\mid \texttt{m}@0 \in \mathbf{Sym}([\![\texttt{cond}]\!]^{\sharp}_{\rho})\},
&
\neg \mathsf{Eq}_{\mathrm A}(R_t,R_e),\\[1mm]
\varnothing, &
\text{otherwise.}
\end{cases}
\]
Here the subscript in \( [\![e]\!]^{\sharp}_{\rho_t} \) indicates that the abstract evaluation of \(e\) is performed under environment \(\rho_t\), and similarly for \(\rho_e\).

\subsection{Detection of non-contributory values}\label{subsec:detection}
After running the above stated static analysis on the host program, 
let \((\rho^\star, C_{\texttt{ctrl}}^\star)\) be the final abstract state produced by the host-side analysis. The set of  symbols that are measurement-originated and may affect the observable host return is
\[
\mathbf{Obs}(\rho^\star, C_{\texttt{ctrl}}^\star)
=
\bigl(\mathbf{Sym}(\rho^\star(\texttt{return})) \cap \{\texttt{m}@0 \mid \texttt{m}\in\mathbf{M}\}\bigr)
\cup
\{\texttt{m}@0 \mid \texttt{m}\in C_{\texttt{ctrl}}^\star\}.
\]
Accordingly, the set of non-contributory measurement-bound variables is
\[
\mathbf{M}_{\mathrm{nc}}
=
\{\texttt{m}\in\mathbf{M}\mid \texttt{m}@0 \notin \mathbf{Obs}(\rho^\star, C_{\texttt{ctrl}}^\star)\}.
\]
\Cref{alg:find-nc} summarizes this procedure.

\subsection{Soundness of Host-side Analysis}
\label{subsec:host-side-analysis-soundness}
We summarize the formal connection between the concrete host semantics and
the abstract analysis; detailed proofs are given in \cref{app:soundness-proofs}.

\paragraph{Concrete and abstract states.}
For the proof, we use instrumented concrete states
\(s=(\nu,\sigma,\delta)\), where \(\nu\) is a valuation of symbols in
\(\Sigma\) respecting \(\mathbf{Dom}\) and \(\mathbf{B}\), \(\sigma\) maps
variables in \(\mathcal V\cup\{\texttt{return}\}\) to concrete runtime values,
and $\delta:\mathcal V\cup\{\texttt{return}\}\to 2^{\mathbf M}$
records the exact set of measurement-bound variables on which each current
runtime value semantically depends. The component \(\delta\) is ghost state used
only in the proof; erasing it yields the ordinary concrete execution. The concrete
collecting domain is \(2^{\mathcal S_{\mathrm c}}\), ordered by set inclusion.

The abstract states are the ones used by the analysis above: a state has the
form \((\rho,C_{\texttt{ctrl}})\), together with the auxiliary maps
\(\mathbf{Dom}\) and \(\mathbf{B}\), where
 $ \rho:\mathcal V\cup\{\texttt{return}\}\to\mathcal A$,
  $C_{\texttt{ctrl}}\subseteq\mathbf M$.
Here \(\mathcal A\) is the abstract-value domain from
\cref{subsec:abs-dom}.

\paragraph{Description relation and concretization.}
For an abstract value \(A\), define its measurement support as
$  \mathbf{MSym}_{\mathbf M}(A)
  =
  \{\,\texttt{m}\in\mathbf M \mid \texttt{m}@0\in\mathbf{Sym}(A)\,\}$.
A concrete value-dependency pair \((v,S_D)\) is described by \(A\) under
\(C_{\texttt{ctrl}}\), written
 $ (v,S_D)\Delta(A,C_{\texttt{ctrl}})$,
if
  $S_D\subseteq
  \mathbf{MSym}_{\mathbf M}(A)\cup C_{\texttt{ctrl}}$,
and additionally, when \(A=\langle p\rangle\) is a polynomial form,
\(v=\nu(p)\). For \(A=\bot\), we require that no concrete return value has been
produced yet and \(S_D=\varnothing\).
We lift \(\Delta\) pointwise to states. We write
\(s\Delta(\rho,C_{\texttt{ctrl}})\) iff for every
\(\texttt{x}\in\mathcal V\cup\{\texttt{return}\}\),
  $(\sigma(\texttt{x}),\delta(\texttt{x}))
  \Delta
  (\rho(\texttt{x}),C_{\texttt{ctrl}})$.
The concretization of an abstract state is therefore
$  \gamma(\rho,C_{\texttt{ctrl}})
  =
  \{\,s\in\mathcal S_{\mathrm c}
      \mid
      s\Delta(\rho,C_{\texttt{ctrl}})\,\}$.
We order abstract states by the concrete states they describe:
\[
  (\rho_1,C_{\texttt{ctrl}}^1)\sqsubseteq_{\gamma}
  (\rho_2,C_{\texttt{ctrl}}^2)
  \quad\Longleftrightarrow\quad
  \gamma(\rho_1,C_{\texttt{ctrl}}^1)
  \subseteq
  \gamma(\rho_2,C_{\texttt{ctrl}}^2).
\]
Thus smaller concretizations are more precise.

\paragraph{Transfer soundness.}
Let \(F_{S_b}\) be the concrete transfer of a statement block \(S_b\):
  $F_{S_b}(X)
  =
  \{\,s'\mid \exists s\in X:\ s\xrightarrow{\ S_b\ }s'\,\}$.
If the abstract transfer defined in \cref{subsec:abs-effect} derives
  $(\rho,C_{\texttt{ctrl}})
  \xRightarrow{\ S_b\ }
  (\rho',C_{\texttt{ctrl}}')$,
then
\[
  F_{S_b}\bigl(\gamma(\rho,C_{\texttt{ctrl}})\bigr)
  \subseteq
  \gamma(\rho',C_{\texttt{ctrl}}').
  \tag{AI-S}
\]
The proof is by induction over expressions and statements, using the rules in
\cref{subsec:abs-eval,subsec:abs-effect}; see \cref{app:soundness-proofs}.

Applying \((\mathrm{AI}\text{-}\mathrm{S})\) to the whole host program \(H\)
shows that the final abstract state
\((\rho^\star,C_{\texttt{ctrl}}^\star)\) over-approximates the exact measurement
dependencies of the concrete host return. Therefore
\(\mathbf{Obs}(\rho^\star,C_{\texttt{ctrl}}^\star)\) over-approximates all
measurement-originated symbols that can affect the return, and every
\(\texttt{m}\in\mathbf{M}_{\mathrm{nc}}\) reported by \textsc{find\_nc} is
semantically non-contributory to the host return.

\subsection{Examples}
The following examples illustrate how our static analysis detects non-contributory measurement outcomes. Each example is presented in a step-by-step form: after every program statement, we show, on the comment line beginning with //, the corresponding abstract state produced by our static analysis.
\begin{example}
\(\texttt{q},\texttt{r}\in\mathbf{M}\), i.e., $\texttt{q}$ and $\texttt{r}$ initially hold measurement outcomes from the quantum circuit.
\[
\begin{array}{@{}l@{\qquad}r@{}}
\texttt{def prog(q,r,b):} 
 // \quad\text{Initial: }\rho(\texttt{q})=\langle \texttt{q}@0\rangle, \rho(\texttt{r})=\langle \texttt{r}@0\rangle, \rho(\texttt{b})=\langle \texttt{b}@0\rangle &\\ 
\quad \texttt{x   = q + b}                  
 \quad//  \rho(\texttt{x}) =[\![\texttt{q}+ \texttt{b}]\!]^{\sharp}=\mathrm{norm}(\langle \texttt{q}@0\rangle + \langle \texttt{b}@0\rangle)= \langle \texttt{q}@0 + \texttt{b}@0\rangle                         & \\
\quad \eta\texttt{ = random(0.1, 0.5)}  \quad    
// \rho(\eta) = \langle \texttt{r}\rangle, 
 \mathbf{Dom}(\texttt{r})=\textbf{real},\ \mathbf{B}(\texttt{r})=[0.1,0.5] & \\
\quad \texttt{y   = int($\eta$ * (q - r))}    \quad
// \rho(\texttt{y}) = \langle0\rangle       & \\
\quad \texttt{q   = q + 1}                  \quad
// \rho(\texttt{q})=\langle \texttt{q}@0 + 1\rangle & \\
\quad \texttt{z   = x - q - b}             \quad
// \rho(\texttt{z}) 
= [\![(\texttt{q}@0 + \texttt{b}@0) - (\texttt{q}@0 + 1) - \texttt{b}@0]\!]^{\sharp}  
= \langle -1\rangle & \\
\quad \texttt{return y + z}                 \quad
// \rho^{\star}(\texttt{return}) 
= [\![\langle 0\rangle + \langle -1\rangle ]\!]^{\sharp}
= \langle -1\rangle
\end{array}
\]
$\rho^{\star}(\texttt{return})$ depends on neither \(q@0\) nor \(r@0\).
Thus, both initial values of \(\texttt{q}\) and \(\texttt{r}\) are not contributing to the final result.
\end{example}

\begin{example}
We assume \(\texttt{m}, \texttt{b}\in\mathbf{M}\).
\[
\begin{array}{@{}l@{\qquad}r@{}}
\texttt{def prog(m,b):} & \\
\quad \texttt{u   = (b + 3)/2} \quad
// \rho(\texttt{u}) = \langle \{\texttt{b}@0\} \rangle & \\
\quad \texttt{v   = foo(m)} \quad
// \rho(\texttt{v}) = \langle \{\texttt{m}@0\} \rangle \quad \text{(unknown call $\Rightarrow$ dependence on $\texttt{m}@0$)} & \\
\quad \eta\texttt{ = random(0.1, 0.5)} \quad
// \rho(\eta) = \langle \texttt{r}\rangle,\ \mathbf{Dom}(\texttt{r})=\textbf{real},\ \mathbf{B}(\texttt{r})=[0.1,0.5] & \\
\quad \texttt{t   = int($\eta$ * (m - m))} \quad
// \rho(\texttt{t}) = \langle 0\rangle & \\
\quad \texttt{m   = random(0, 1)} \quad
// \rho(\texttt{m})=\langle \texttt{r}'\rangle,\ \mathbf{Dom}(\texttt{r}')=\textbf{real},\ \mathbf{B}(\texttt{r}')=[0,1] & \\
\quad \texttt{w   = u + t + m} \quad
// \rho(\texttt{w}) = \langle \{\texttt{b}@0,\, \texttt{r}'\} \rangle & \\
\quad \texttt{return w} \quad
// \rho^{\star}(\text{return}) = \langle \{\texttt{b}@0,\, \texttt{r}'\} \rangle & \\
\end{array}
\]
The observable return depends on \(\texttt{b}@0\) and the fresh random value \(\texttt{r}'\), but not on the initial value \(\texttt{m}@0\). So the initial value of \(\texttt{b}\) contributes to the final result, but that of $\texttt{m}$ does not.

\end{example}

\begin{example}
Assume $\texttt{m},\texttt{n}\in\mathbf{M}$.
\[
\begin{array}{@{}l@{\qquad}r@{}}
\texttt{def prog(m,n,b):} & \\
\quad \texttt{u = b + n} \quad
// \rho(\texttt{u})=\langle \texttt{b}@0 + \texttt{n}@0\rangle \\
\quad \texttt{if (m):} \quad
// [\![\texttt{m}]\!]^\sharp=\langle \texttt{m}@0\rangle \\
\qquad \texttt{return u} 
\qquad// \rho_t(\texttt{return})=\langle \texttt{b}@0 + \texttt{n}@0\rangle \\
\quad \texttt{else:} & \\
\qquad \texttt{return u + 1} 
\qquad// \rho_e(\texttt{return})=\langle \texttt{b}@0 + \texttt{n}@0 + 1\rangle &\\
// \rho^\star=\rho_t\sqcup_{\mathrm{env}}\rho_e, C_{\texttt{ctrl}}^{\star} = \{\} \cup \mathbf{Sym}([\![\texttt{m}]\!]^{\sharp})=\{\texttt{m}@0\}
\end{array}
\]
Because $\rho_e(\texttt{return}) \not\equiv \rho_t(\texttt{return})$, the branch changes the returned value, the post-conditional environment joins the two branch returns conservatively.
So, the final return depends on $\texttt{b}@0$ and $\texttt{n}@0$; besides, although $\texttt{m}@0$ does not appear in $\rho^\star(\texttt{return})$, it is in $C^\star_{\texttt{ctrl}}$, so $\texttt{m}$ is also contributory.

\end{example}

\subsection{Integration with circuit-level simplification}\label{subsec:integration}
$\mathbf{M}_{\mathrm{nc}}$ denotes the set of measurement outcomes identified as non-contributory by our analysis.
In this section, we integrate the analysis into the circuit simplification pipeline so that every gate acting solely within the causal cones of $\mathbf{M}_{\mathrm{nc}}$ is safely removed.
$Q$ denotes the set of circuit qubits, and the measurement map
$\mu:Q\rightharpoonup\mathcal{V}$ maps each measured qubit to the host variable that stores its outcome.
For a measured qubit $q$, its backward causal cone is the set of gates and wire segments that can reach the terminal measurement of $q$ by following the circuit forward through wires and multi-qubit gates. A gate whose forward influence reaches only qubits mapped to $\mathbf{M}_{\mathrm{nc}}$ is dead.

The workflow is \cref{alg:host-aware-simpl}. 
$\textsc{find\_nc}$ runs the host-side analysis to compute $\mathbf{M}_{\mathrm{nc}}$, whose elements we map back to their qubits and treat as dead in the DGE sense. 
Any gate whose effects propagate only to dead qubits cannot change observable behaviour, so we sweep the circuit, delete such gates, and grow the dead-qubit set backwards along causal cones.

\subsection{Asymptotic analysis}\label{subsec:asymptotic-analysis}
\cref{corollary:efficiency} shows that, for typical hybrid programs with bounded abstract value sizes, our static analysis scales linearly in program size, enabling fast optimization on large codes. 
\cref{corollary:efficiency-workflow} provides the corresponding efficiency guarantee for the host-aware circuit optimization in \cref{alg:host-aware-simpl}, and \cref{theorem:complexity} states its worst-case time complexity of our static analysis.

Let $n$ be the number of operators in the host program, counting arithmetic, \texttt{int}$(\cdot)$, \texttt{random}$(\cdot,\cdot)$, conditionals, and function calls.
Let $|\mathcal{V}|$ be the number of variables tracked in the abstract environment, $\mathcal{N}$ the maximum number of monomials appearing in any polynomial-form abstract value, and $|D|$ the maximum size of any dependence-form symbol set. Detailed proofs are in \cref{app:asym-proofs}.

\begin{theorem}\label{theorem:complexity}
In the worst case, the proposed host-side static analysis runs in 
$\mathcal{O}\!\left(n\cdot\mathcal{N}^2 + n\cdot|\mathcal{V}|\cdot|D|\right)$.

\end{theorem}

\begin{corollary}\label{corollary:efficiency}
When $\max\big(\ \mathcal{N}, |\mathcal{V}|, |D|\ \big)$ is bounded by a constant value, the worst-case time complexity of the static analysis on the host program is $\mathcal{O} \big(n \big)$.

\end{corollary}

\begin{corollary}\label{corollary:efficiency-workflow}
Let $G$ be the number of gates in the input quantum circuit. The worst-case time complexity of the workflow in \cref{alg:host-aware-simpl} is $\mathcal{O}(n\cdot\mathcal{N}^2 +n\cdot|\mathcal{V}|\cdot|D| + G^2)$, and when  $\max\big(\ \mathcal{N}, |\mathcal{V}|, |D|\ \big)$ is bounded by a constant value, the complexity becomes $\mathcal{O}(n + G^2)$.
\end{corollary}

\subsection{Levelized SSA for GPU Acceleration}\label{subsec:method-gpu-acc}

\cref{alg:host-aware-simpl} invokes \textsc{find\_nc} to compute $\mathbf{M}_{\mathrm{nc}}$, 
which, in this subsection, we discuss how to accelerate with GPUs. However, the challenge is, as confirmed later in \hyperref[par:why-lowering-matters]{the last report} in \cref{subsec:gpu-acc}, that directly applying GPU acceleration to the original host-side static analysis brings no practical benefit.

To expose parallel work, we lower the host program $H$ to an
SSA-style, levelized intermediate representation
$\widehat{H} = (\mathcal{L}_0,\mathcal{L}_1,\ldots,\mathcal{L}_{K-1})$,
where each instruction writes a fresh destination and each level $\mathcal{L}_k$ contains instructions that can be evaluated in static analysis in parallel.
Writing $\mathrm{Def}(\mathcal{L}_k)$ for the set of destinations
defined in level $\mathcal{L}_k$, and $\mathrm{Use}(s)$ for the
operands read by instruction $s$, the lowering satisfies
$\forall s \in \mathcal{L}_k:\mathrm{Use}(s) \subseteq \Sigma_0 \cup \bigcup_{j<k}\mathrm{Def}(\mathcal{L}_j)$.
Hence, two instructions in the same level never have read-after-write, write-after-read, or write-after-write hazards.

The host program $H$ is lowered in two steps.
First, it performs a semantics-preserving SSA/three-address expansion. Each assignment or return expression is decomposed into primitive operators already supported by our transfer functions, namely binary arithmetic, \texttt{int}$(\cdot)$, \texttt{random}$(\cdot,\cdot)$, conservative unknown calls, and branch merges. 
Every definition is rewritten to a fresh destination name.
Nested expressions are flattened so that each lowered instruction contains only one primitive operator, and the final returned value is still written to the distinguished variable $\texttt{return}$. 
For a structured conditional, we lower the two branches recursively and make the merge explicit by an instruction of the form
$u := \phi_{\mathrm{if}}(c,v_t,v_e)$,
where $c$ is the value of the condition, and $v_t,v_e$ are the values produced by the then/else branches.

Second, the resulting SSA program is levelized.
Let $\lambda(v)$ denote the level at which the SSA value $v$ becomes available.
For any SSA value whose current SSA version is not defined by another lowered instruction we set $\lambda(v)=-1$.
For each lowered instruction $s$, we then assign its level $\ell(s)$ by the earliest-ready rule
$\ell(s)=
1+\max\bigl(\{-1\}\cup\{\lambda(v)\mid v\in \mathrm{Use}(s)\}\bigr)$,
and define
$\mathcal{L}_k=\{\,s \mid \ell(s)=k\,\}$.

To make the levelized lowering concrete, see \cref{ex:lower-to-level-ssa}.
\begin{example}\label{ex:lower-to-level-ssa}
Consider the following host program $H$, where \(\texttt{q}, \texttt{r}, \texttt{b}\in\mathbf{M}\):
\[
\begin{array}{@{}l@{\qquad}r@{}}
\texttt{def prog(q,r,b):} & \\
\quad \texttt{x = q + b} & \\
\quad \texttt{eta = random(0.1, 0.5)} & \\
\quad \texttt{y = int(eta * (q - r))} & \\
\quad \texttt{q = q + 1} & \\
\quad \texttt{z = x - q - b} & \\
\quad \texttt{return y + z} & \\
\end{array}
\]
Lower this program to
$\widehat H = (\mathcal L_0,\mathcal L_1,\mathcal L_2,\mathcal L_3,\mathcal L_4)$.
\(\texttt q_0,\texttt r_0,\texttt b_0\) denote the incoming versions of the source variables,
and each assignment $s_i$ writes a fresh destination:
{\small
\[
\begin{aligned}
\mathcal L_0 = \{ &
s_1: x_1 := q_0 + b_0,
s_2: \eta_1 := \texttt{random}(0.1,0.5),
s_3: q_1 := q_0 + 1,
s_4: t_1 := q_0-r_0
\},\\[1mm]
\mathcal L_1 = \{ &
s_5: p_1 := \eta_1 \cdot t_1,\;
s_6: z_1 := x_1 - q_1
\},
\mathcal L_2 = \{ 
s_7: y_1 := \texttt{int}(p_1),\;
s_8: z_2 := z_1 - b_0
\},\\[1mm]
\mathcal L_3 = \{ &
s_9: ret_1 := y_1 + z_2
\,\},
\mathcal L_4 = \{\
s_{10}: \texttt{return} := ret_1
\,\}.
\end{aligned}
\]
}
The initial values here are the same as in those defined in \cref{subsec:abs-dom}, i.e., 
$\rho(q_0)=\langle q@0\rangle$,
$\rho(r_0)=\langle r@0\rangle$,
$\rho(b_0)=\langle b@0\rangle$.
The new names \(x_1,\eta_1,q_1,t_1,\ldots\) are only fresh SSA
destinations used by the implementation to avoid overwriting
abstract values in place.
For this lowered program, some level-definition sets are
$\mathrm{Def}(\mathcal L_0)=\{x_1,\eta_1,q_1,t_1\}$,
$\mathrm{Def}(\mathcal L_1)=\{p_1,z_1\}$,
and some use-sets are
$\mathrm{Use}(s_1)=\{q_0,b_0\}$,
$\mathrm{Use}(s_2)=\emptyset$,
$\mathrm{Use}(s_5)=\{\eta_1,t_1\}$.
Here \(\mathrm{Use}(s)\) could be read as the set of operands that must
already have values available before instruction \(s\) can be
evaluated. For example, \(s_5\) cannot be placed in \(\mathcal L_0\),
as it needs \(\eta_1\) and \(t_1\), both of which are only defined
after \(\mathcal L_0\). 
\end{example}

The GPU backend processes one level at a time:
for each $\mathcal{L}_k$, it launches $|\mathcal{L}_k|$ threads,
one thread per instruction, and synchronizes only between consecutive levels.
As destinations within a level are disjoint and every source operand comes from an earlier level, the result is independent of the intra-level execution order.

\textbf{Preservation of the lowered analysis.}
The GPU backend changes the representation and schedule of the host-side
analysis, but it should not change the analysis result. \Cref{app:gpu-proofs} gives more details in showing that the above lowering is semantics-preserving with respect to our host-side analysis.

\section{Evaluation}\label{sec:evaluation}
Our evaluation serves two complementary goals. 
\textbf{(i) Effectiveness in \cref{subsec:effect}}:
we evaluate whether the implemented host-aware optimization exposes circuit reductions missed by circuit-only optimizers.
We first run $3$ diagnostic workloads as smoke tests, then measure $24$ upstream-artifact-driven workloads across quantum chemistry, quantum optimization, quantum machine learning, and quantum finance, with four controls checking the expected dead-set behavior.
We compare our pass against Qiskit, t$\ket{\text{ket}}$, and PyZX, and compose it before and after each optimizer, to quantify the benefit beyond reach of optimizers not taking host programs into account.
We also compare against a syntactic host-liveness baseline, which distinguishes dead measurements detectable by syntactic liveness from those exposed only by our semantic reasoning.
\textbf{(ii) GPU acceleration in \cref{subsec:gpu-acc}}:
we isolate the host-side static analysis $\textsc{find\_nc}(\cdot,\cdot)$ and compare a sequential baseline against a CUDA backend on controlled synthetic host programs with tunable structural parallelism.
In addition, we verify that directly accelerating the original host-side static analysis on GPUs is ineffective, motivating the levelized-SSA lowering introduced in \cref{subsec:method-gpu-acc}.

\subsection{Effectiveness}\label{subsec:effect}

\hypertarget{passsettings}{\textbf{Settings.}}
We evaluate whether the host-side analysis exposes optimizations that are missed by circuit-only optimizers.  We use the following pipelines:
\begin{itemize}
    \item Pipeline \textbf{q}: Qiskit's \texttt{transpile} at optimization level~3, using its default target-agnostic configuration.
    \item Pipeline \textbf{t}: t$\ket{\text{ket}}$ backend-agnostic optimization sequence
    \texttt{DecomposeBoxes} $\rightarrow$
    \texttt{PeepholeOptimise2Q} $\rightarrow$
    \texttt{KAKDecomposition} $\rightarrow$
    \texttt{CommuteThroughMultis} $\rightarrow$
    \texttt{CliffordSimp} $\rightarrow$
    \texttt{RemoveRedundancies} $\rightarrow$
    \texttt{FullPeepholeOptimise}.  We invoke \texttt{CliffordSimp} with qubit permutations/SWAP introduction disabled when supported, so that this pipeline remains hardware agnostic.
    \item Pipeline \textbf{p}: PyZX's \texttt{full\_optimize}, which converts the circuit to a ZX graph, applies ZX rewrites, and extracts a simplified circuit.
    \item Pipeline \textbf{h}: our host-aware optimization.  It first runs the host-side dependency analysis to compute which measurement variables can affect the host return value.  Measurements outside this dependency set are mapped back to quantum wires, and the corresponding dead backward cones are removed.
    \item Pipeline \textbf{l}: a syntactic host-liveness baseline.  Following the standard backward live-variable/data-flow formulation for imperative programs~\cite{alfred2007compilers,kildall1973unified}, it starts from the host return expression and propagates liveness backward through assignments and branch guards.  
    The baseline is control-dependence-aware, but purely syntactic: it does not simplify arithmetic expressions, zero coefficients, equivalent branches, or bounded \texttt{int}/\texttt{floor} residuals. \cref{alg:syntactic-live} summarizes its implementation.
    \item Pipelines \textbf{qh}, \textbf{hq}, \textbf{th}, \textbf{ht}, \textbf{ph}, and \textbf{hp}: compositions of the above pipelines, where the concatenation ``XY'' means ``run X, then Y''.
\end{itemize}

\noindent\textbf{Smoke test.} We first run our method on three small but realistically designed diagnostic workloads. Its setup and results are in \cref{app:diagnostic-and-savings}.

The main effectiveness evaluation uses $24$ upstream-artifact-driven workloads, six each from quantum chemistry, quantum optimization, quantum machine learning, and quantum finance.  Each main workload is paired with a workload-specific quantum kernel derived from the MQTBench benchmark templates \cite{Quetschlich2023MQTBench} and stored as an OpenQASM.  
The host program consumes measurement variables produced by the quantum kernel and computes an application-level quantity, such as an energy estimate, a QUBO objective value, a classifier score, or a financial payoff/risk quantity.
The upstream artifacts model application-side preprocessing outputs like
screened Hamiltonian terms, symmetry-sector certificates, QUBO presolve reports, sparse trained-model weights, portfolio exposures, and risk-reporting
policies.  
A detailed description of the workloads is in \cref{app:workloads}.
These artifacts define the host-side computation, and our host-side analysis computes which measurement variables are semantically irrelevant to the final host return value, driving DGE to prune quantum circuits.

We additionally run $4$ control workloads.  
These controls are excluded from aggregate effectiveness numbers in \cref{tab:effect-aggregate}. They are used only to validate the analysis: two all-live negative controls, one positive control with an expected screened chemistry tail, and one zero-weight trap control in which a locally zero coefficient does not make the corresponding measurement dead because the same feature is still used by an auxiliary live term.

Before recording metrics, each resulting circuit is canonicalized to the common basis
$\{\texttt{rz},\texttt{sx},\texttt{x},\texttt{cx}\}$.
Here \texttt{rz} is a rotation around the $Z$ axis, \texttt{sx} is the square-root of the Pauli-$X$ gate, \texttt{x} is the Pauli-$X$ bit flip, and \texttt{cx} is controlled-$X$/CNOT. We use this basis only to canonicalize gate counts across toolchains, not as a hardware target model.
For each resulting circuit, we record the number of one-qubit gates, two-qubit gates, and circuit depth, excluding terminal measurements from the gate counts and depth.  We also record total-gate percentage reduction, where total gates are defined as the sum of one- and two-qubit gates.  The experiments run with Python~3.8.10 under WSL2 with Qiskit~0.46.3, \texttt{pytket}~1.3.0, PyZX~0.7.3.

\noindent\textbf{Results.}
\Cref{tab:effect-aggregate} summarizes the aggregate total-gate reductions on the $24$ main workloads.  Running our host-aware pass alone reduces total gate
count by $38.0\%$ on average.  More importantly, the host-side analysis remains effective when composed with standard circuit-only optimizers.  When applied
after the circuit optimizer, it yields an additional $31.6\%$ mean total-gate
reduction after Qiskit, $31.1\%$ after t$\ket{\text{ket}}$, and $31.1\%$ after
PyZX.  We also observe essentially the same reductions when the host-aware pass
is applied before Qiskit or t$\ket{\text{ket}}$ and the resulting circuit is
then optimized, with means of $31.6\%$ and $31.1\%$, respectively.  For PyZX, the pre-optimization composition is slightly less favorable but still gives a $30.8\%$ mean reduction relative to PyZX alone.  

\begin{table}[htb]
\centering
\small
\setlength{\tabcolsep}{4pt}
\caption{Aggregate total-gate reductions on the $24$ main effectiveness workloads. 
}
\label{tab:effect-aggregate}
\begin{tabular}{lcc}
\toprule
Comparison & Mean total-gate reduction & Range \\
\midrule
\textbf{original} $\rightarrow$ \textbf{h} & 37.98\% & 15.67--62.88\% \\
\textbf{q} $\rightarrow$ \textbf{qh} & 31.61\% & 3.88--65.62\% \\
\textbf{q} $\rightarrow$ \textbf{hq} & 31.61\% & 3.88--65.62\% \\
\textbf{t} $\rightarrow$ \textbf{th} & 31.09\% & 2.55--68.87\% \\
\textbf{t} $\rightarrow$ \textbf{ht} & 31.09\% & 2.55--68.87\% \\
\textbf{p} $\rightarrow$ \textbf{ph} & 31.10\% & 3.94--69.45\% \\
\textbf{p} $\rightarrow$ \textbf{hp} & 30.75\% & 0.73--69.45\% \\
\bottomrule
\end{tabular}
\end{table}

\Cref{fig:effect-lollipop} breaks down the post-optimizer reductions per workload.  The benefits appear across all four application families rather than being concentrated in a single domain.  
Chemistry and finance workloads exhibit more moderate but still consistently positive reductions, reflecting smaller or more localized dead backward cones than other workloads.

\begin{figure}[htb]
  \centering
  \includegraphics[width=\textwidth]{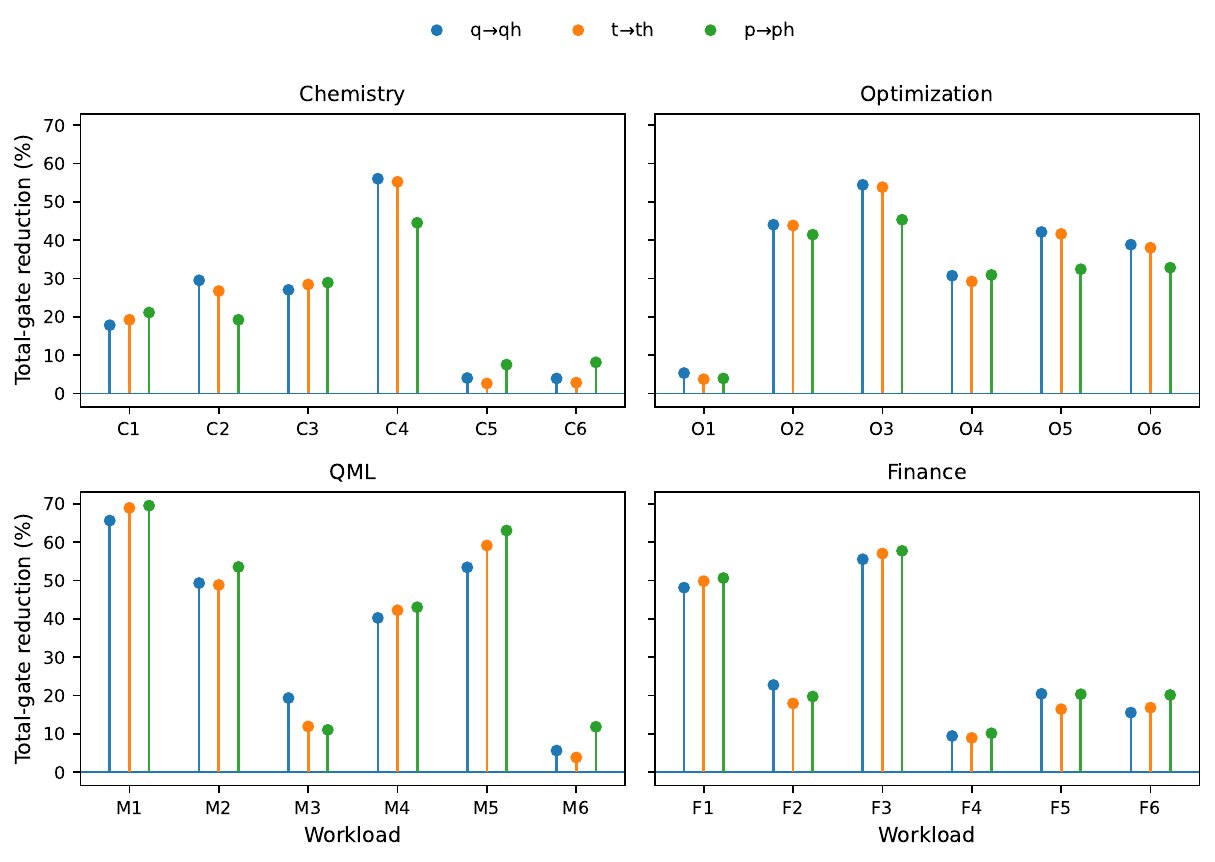}
  \caption{Per-workload total-gate reduction for the main effectiveness workloads.  For each workload, the figure reports the additional total-gate reduction obtained by composing host-side analysis with Qiskit (\textbf{q}$\rightarrow$\textbf{qh}), t$\ket{\text{ket}}$ (\textbf{t}$\rightarrow$\textbf{th}), and PyZX (\textbf{p}$\rightarrow$\textbf{ph}).  Workloads are grouped by application family.}
  \label{fig:effect-lollipop}
\end{figure}

As a side check closer to the novelty of the host analysis, we also compare against the syntactic liveness baseline \hyperlink{passsettings}{\textbf{l}}.  
To isolate the contribution of the host-side dead-measurement discovery from the circuit-pruning backend, we feed the dead measurements found by \hyperlink{passsettings}{\textbf{l}} and by our semantic analysis \hyperlink{passsettings}{\textbf{h}} into the same DGE backend.  
\Cref{tab:liveness-baseline} reports both the measurement-level dead sets and the resulting circuit-level reductions.  The baseline \hyperlink{passsettings}{\textbf{l}} identifies $21$ dead measurement variables and reduces total gate count by $5.06\%$ on average, while \hyperlink{passsettings}{\textbf{h}} identifies $92$ dead measurement variables and reduces total gate count by $37.98\%$ on average under the same DGE backend.  Hence $71$ dead measurements are semantic-only opportunities not recovered by the syntactic baseline.

\begin{table}[t]
\centering
\small
\setlength{\tabcolsep}{3pt}
\renewcommand{\arraystretch}{1.08}
\caption{Comparison between the syntactic host-liveness baseline \textbf{l} and the semantic host-side analysis \textbf{h} on the $24$ main workloads. Syn.-dead counts measurement variables detected by syntactic liveness; Sem.-dead counts measurement variables detected by our semantic analysis; Sem.-only is the additional dead set found only by semantic reasoning. DGE-\textbf{l} and DGE-\textbf{h} report mean total-gate reduction after feeding the corresponding dead-measurement set into the same DGE backend, starting from the original circuit.}
\begin{tabular}{lrrrrrrr}
\toprule
Domain & Workloads & Meas. & Syn.-dead & Sem.-dead & Sem.-only & DGE-\textbf{l} & DGE-\textbf{h} \\
\midrule
Chemistry & 6 & 62 & 17 & 25 & 8 & 11.05\% & 30.50\% \\
Optimization & 6 & 49 & 4 & 21 & 17 & 9.20\% & 47.03\% \\
QML & 6 & 65 & 0 & 23 & 23 & 0.00\% & 40.83\% \\
Finance & 6 & 73 & 0 & 23 & 23 & 0.00\% & 33.58\% \\
\midrule
\textbf{Total} & 24 & 249 & 21 & 92 & 71 & 5.06\% & 37.98\% \\
\bottomrule
\end{tabular}
\label{tab:liveness-baseline}
\end{table}
As shown in \cref{tab:control-validation}, both all-live negative controls produce an empty dead-measurement set and zero host-induced savings.  The positive control recovers exactly the expected screened measurements, while the zero-weight trap control also produces an empty dead set because the zero-weighted feature remains live through an auxiliary term.

\begin{table}[htb]
\caption{Control-workload validation.  Negative controls are expected to produce no dead measurements and no host-induced pruning.  The positive control checks that the analysis can recover a known screened dead set.}
\centering
\small
\setlength{\tabcolsep}{4pt}
\renewcommand{\arraystretch}{1.08}
\begin{tabular}{llcccc}
\toprule
Ctrl. & Role & Expected & Observed & Dead-set & No-pruning \\
\midrule
\texttt{CTRL1} & negative & $\emptyset$ & $\emptyset$ & pass & pass \\
\texttt{CTRL2} & negative & $\emptyset$ & $\emptyset$ & pass & pass \\
\texttt{CTRL3} & positive & $\{\texttt{m}_2,\texttt{m}_8,\texttt{m}_9,\texttt{m}_{10}\}$ & $\{\texttt{m}_2,\texttt{m}_8,\texttt{m}_9,\texttt{m}_{10}\}$ & pass & n/a \\
\texttt{CTRL4} & negative trap & $\emptyset$ & $\emptyset$ & pass & pass \\
\bottomrule
\end{tabular}
\label{tab:control-validation}
\end{table}

Overall, these results show that host-side semantic information is complementary to circuit-only optimization.  
Circuit optimizers can simplify the quantum circuit structure, but they do not know which measured values are semantically consumed by the host computation.  
Our analysis exposes this information and enables additional dead-gate elimination even after Qiskit, t$\ket{\text{ket}}$, and PyZX have optimized the same circuits.
Combined with the result from \cref{app:diagnostic-and-savings}, the best composition order is workload-dependent: in some cases applying \hyperlink{passsettings}{\textbf{h}} after a
circuit-only optimizer is better, in some cases applying \hyperlink{passsettings}{\textbf{h}} first exposes more downstream simplification, and in other cases the order has little or no effect.

\subsection{GPU acceleration of host-side analysis}\label{subsec:gpu-acc}
Our goal in this subsection is to quantify the practical effect of GPU acceleration on the host-side static analysis \textsc{find\_nc}$(\cdot,\cdot)$ (\cref{alg:find-nc}).
To this end, we compare a sequential baseline against a GPU backend over a family of synthetic host programs whose size, source-level parallel structure, and symbolic coupling can be varied independently.

\noindent\textbf{Setting.}
Each generated input is an imperative host program with exactly $|\mathbf{M}|$ measurement-bound inputs and $s_h$ source statements.
As illustrated in \cref{fig:generated-input-host-program-structure}, 
generation is organized around \emph{active chains}, each of which maintains its own current frontier value and advances mostly independently before the final reduction to \texttt{return}.
A structural-parallelism knob $p \in [0,1]$ determines the number of chains.
Given $|\mathbf{M}|$ and $s_h$, we compute
$n_{\max}=
\min\!\left(
|\mathbf{M}|,
\left\lfloor \frac{s_h+1}{2} \right\rfloor
\right),
n_{\mathrm{chain}}=
1+\left\lfloor p\,(n_{\max}-1)\right\rfloor $.
This reserves $n_{\mathrm{chain}}-1$ statements for the final reduction and leaves
$s_{\mathrm{body}} = s_h - (n_{\mathrm{chain}}-1)$
statements for the body.
Round by round, each chain with remaining budget emits one statement that overwrites one variable in its small working set.
Statements are sampled from the primitive operator family in \cref{sec:methods}; each non-constant operand is drawn from the same chain with probability $1-\chi$ and from a different chain with probability $\chi$, where $\chi \in [0,1]$ controls cross-chain reads.

The sequential baseline analyzes the original host program directly on CPU, in source order, using a serial implementation of the host-side analysis.
The GPU backend first lowers the host program on the CPU to the levelized SSA representation of \cref{subsec:method-gpu-acc} and then executes the resulting parallel analysis tasks on the GPU using CUDA.
Both executables are compiled with \texttt{nvcc}~12.8.93.
All experiments are run on the same node, where two NVIDIA H100 NVL devices are visible to the process.
For every design point, both implementations analyze exactly the same program.
Each point is instantiated by one deterministic source program from a fixed seed and timed over $10$ repetitions after $3$ warmup runs.
The reported running time of the GPU backend includes both the CPU-side lowering time and the GPU execution time.

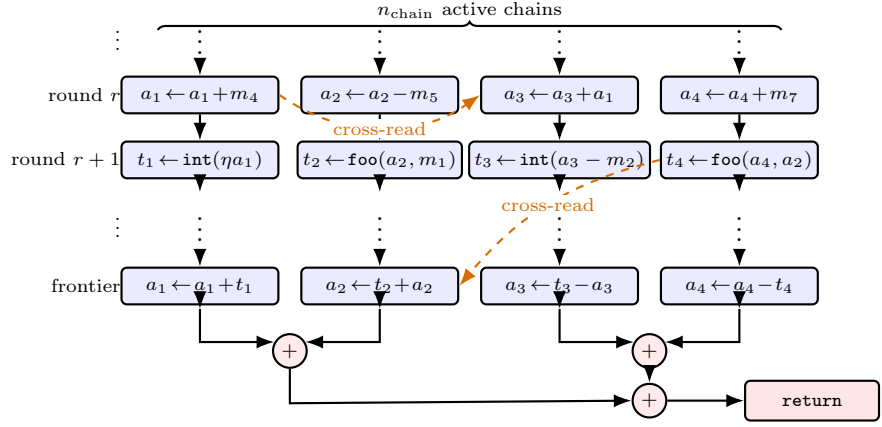
\begin{figure}[htb]
\centering
\resizebox{0.95\columnwidth}{!}{%
\begin{tikzpicture}[
    >=Latex,
    font=\scriptsize,
    stmt/.style={
        draw,
        rounded corners=2pt,
        thick,
        align=center,
        fill=blue!8,
        minimum width=2.05cm,
        minimum height=0.48cm,
        inner xsep=2pt,
        inner ysep=0.8pt
    },
    merge/.style={
        draw,
        circle,
        thick,
        fill=red!8,
        minimum size=4.4mm,
        inner sep=0pt
    },
    retbox/.style={
        draw,
        rounded corners=2pt,
        thick,
        align=center,
        fill=red!10,
        minimum width=1.75cm,
        minimum height=0.52cm,
        inner xsep=2pt,
        inner ysep=0.9pt
    },
    lab/.style={font=\scriptsize, align=center},
    flow/.style={-Latex, thick},
    cross/.style={-Latex, thick, dashed, orange!85!black},
    dotslab/.style={font=\normalsize}
]

\def\xA{0.00}
\def\xB{2.35}
\def\xC{4.70}
\def\xD{7.05}

\def\yTop{-0.35}
\def\yR{-1.20}
\def\yRp{-2.05}
\def\yMid{-2.85}
\def\yF{-3.70}
\def\yFan{-4.00}
\def\yM{-4.55}
\def\yRet{-5.20}

\draw[decorate,decoration={brace,amplitude=3pt},thick]
(-0.55,\yTop) -- (7.60,\yTop)
node[midway,above=1pt] {$n_{\mathrm{chain}}$ active chains};

\node[lab, anchor=east] at (-0.92,\yTop-0.05) {$\vdots$};
\node[lab, anchor=east] at (-0.92,\yR) {round $r$};
\node[lab, anchor=east] at (-0.92,\yRp) {round $r+1$};
\node[lab, anchor=east] at (-0.92,\yMid) {$\vdots$};
\node[lab, anchor=east] at (-0.92,\yF) {frontier};

\node[dotslab] (d1t) at (\xA,\yTop-0.05) {$\vdots$};
\node[dotslab] (d2t) at (\xB,\yTop-0.05) {$\vdots$};
\node[dotslab] (d3t) at (\xC,\yTop-0.05) {$\vdots$};
\node[dotslab] (d4t) at (\xD,\yTop-0.05) {$\vdots$};

\node[stmt] (c11) at (\xA,\yR) {$a_1\!\leftarrow\!a_1\!+\!m_4$};
\node[stmt] (c21) at (\xB,\yR) {$a_2\!\leftarrow\!a_2\!-\!m_5$};
\node[stmt] (c31) at (\xC,\yR) {$a_3\!\leftarrow\!a_3\!+\!a_1$};
\node[stmt] (c41) at (\xD,\yR) {$a_4\!\leftarrow\!a_4\!+\!m_7$};

\node[stmt] (c12) at (\xA,\yRp) {$t_1\!\leftarrow\!\texttt{int}(\eta a_1)$};
\node[stmt] (c22) at (\xB,\yRp) {$t_2\!\leftarrow\!\texttt{foo}(a_2,m_1)$};
\node[stmt] (c32) at (\xC,\yRp) {$t_3\!\leftarrow\!\texttt{int}(a_3-m_2)$};
\node[stmt] (c42) at (\xD,\yRp) {$t_4\!\leftarrow\!\texttt{foo}(a_4,a_2)$};

\node[dotslab] (d1m) at (\xA,\yMid) {$\vdots$};
\node[dotslab] (d2m) at (\xB,\yMid) {$\vdots$};
\node[dotslab] (d3m) at (\xC,\yMid) {$\vdots$};
\node[dotslab] (d4m) at (\xD,\yMid) {$\vdots$};

\node[stmt] (c13) at (\xA,\yF) {$a_1\!\leftarrow\!a_1\!+\!t_1$};
\node[stmt] (c23) at (\xB,\yF) {$a_2\!\leftarrow\!t_2\!+\!a_2$};
\node[stmt] (c33) at (\xC,\yF) {$a_3\!\leftarrow\!t_3\!-\!a_3$};
\node[stmt] (c43) at (\xD,\yF) {$a_4\!\leftarrow\!a_4\!-\!t_4$};

\draw[flow] (d1t.south) -- (c11.north);
\draw[flow] (c11.south) -- (c12.north);
\draw[flow] (c12.south) -- (d1m.north);
\draw[flow] (d1m.south) -- (c13.north);

\draw[flow] (d2t.south) -- (c21.north);
\draw[flow] (c21.south) -- (c22.north);
\draw[flow] (c22.south) -- (d2m.north);
\draw[flow] (d2m.south) -- (c23.north);

\draw[flow] (d3t.south) -- (c31.north);
\draw[flow] (c31.south) -- (c32.north);
\draw[flow] (c32.south) -- (d3m.north);
\draw[flow] (d3m.south) -- (c33.north);

\draw[flow] (d4t.south) -- (c41.north);
\draw[flow] (c41.south) -- (c42.north);
\draw[flow] (c42.south) -- (d4m.north);
\draw[flow] (d4m.south) -- (c43.north);

\draw[cross]
(c11.east) to[bend right=35]
node[midway, fill=white, inner sep=1pt] {cross-read}
(c31.west);

\draw[cross]
(c42.west) to[bend right=18]
node[midway, fill=white, inner sep=1pt] {cross-read}
(c23.east);

\coordinate (f1) at (\xA,\yFan);
\coordinate (f2) at (\xB,\yFan);
\coordinate (f3) at (\xC,\yFan);
\coordinate (f4) at (\xD,\yFan);

\draw[flow] (c13.south) -- (f1);
\draw[flow] (c23.south) -- (f2);
\draw[flow] (c33.south) -- (f3);
\draw[flow] (c43.south) -- (f4);

\node[merge] (m12) at (1.175,\yM) {$+$};
\node[merge] (m34) at (5.875,\yM) {$+$};

\node[merge] (mret) at (5.875,\yRet) {$+$};
\node[retbox] (ret) at (8.00,\yRet) {\texttt{return}};

\draw[flow] (f1) |- (m12.150);
\draw[flow] (f2) |- (m12.30);

\draw[flow] (f3) |- (m34.150);
\draw[flow] (f4) |- (m34.30);

\draw[flow] (m34.south) -- (mret.north);
\draw[flow] (m12.south) -- ++(0,-0.45) -- (mret.west);

\draw[flow] (mret.east) -- (ret.west);

\end{tikzpicture}%
}
\caption{Illustration of active chains in input programs.}
\label{fig:generated-input-host-program-structure}
\end{figure}

\noindent\textbf{Results.}
Across all tested design points, the CUDA backend produced the same analysis result as the sequential baseline.
We first summarize the speedup landscape in (a) of \cref{fig:gpu-speedup-combined}.
Using the experimental setup described above, we sweep
$|\textbf{M}|\in\{32,64,128,256\}$ and
$s_h\in\{1024,2048,4096,8192\}$, and evaluate the  speedup as $T_{\mathrm{seq}}/T_{\mathrm{gpu}}$
for several representative source-level parallelism settings.
The trend is clearly visible:
GPU acceleration becomes stronger as either the number of measurement-bound inputs $|\textbf{M}|$ or the host-program size $s_h$
increases, and this trend is amplified by larger $p$.
This directly supports the intuition that, once lowering exposes sufficient structural parallelism in the host-side analysis, the analysis admits substantial parallel execution.

We next make the break-even behavior explicit in (b) of \cref{fig:gpu-speedup-combined}.
This figure uses the same experimental sweep, but now keeps a larger set of parallelism settings $p\in\{0,0.125,0.25,0.5,0.75,1.0\}$ and groups the results by the source host size $s_h$.
Each cell is still annotated by the measured  speedup,
but the background color indicates whether the GPU backend wins
($\mathrm{speedup}>1$) or loses ($\mathrm{speedup}<1$).
The resulting phase diagram shows a clear break-even boundary.
For the smallest host size ($s_h=1024$), the GPU backend requires moderate structural parallelism to become worthwhile.
As the source host size increases, this boundary shifts leftward.
By $s_h=4096$ or $8192$, the GPU backend already wins on all tested combinations.

\phantomsection
\label{par:why-lowering-matters}
\noindent\textbf{Why lowering matters.}
We also examine what happens if the host-side analysis is moved to the GPU directly on the original host program, without first applying the levelized SSA lowering of \cref{subsec:method-gpu-acc}.
We compare two implementations of the same host-side analysis:
the first, which we call the \emph{source-order CUDA implementation}, analyzes the original host program directly on GPUs;
the second remains the sequential baseline we used in our previous experiments. For each design point $(|\textbf{M}|, s_h)$, we try all tested values of the requested structural parallelism $p$ and record the best speedup over the sequential baseline observed for that design point.
Across all tested design points, the source-order CUDA implementation remains far below break-even: even the best observed speedup lies only in the range $0.026\times$--$0.035\times$.
This is expected: in its original form, the analysis exposes too little regular parallel work for SIMT execution, so kernel-launch overheads, irregular memory access, and synchronization costs dominate the useful computation.
Thus, useful GPU acceleration requires the lowering step of \cref{subsec:method-gpu-acc} to expose additional parallelism first.

\begin{figure}[t]
    \centering
    \begin{minipage}[t]{0.48\columnwidth}
        \centering
        \includegraphics[width=\linewidth]{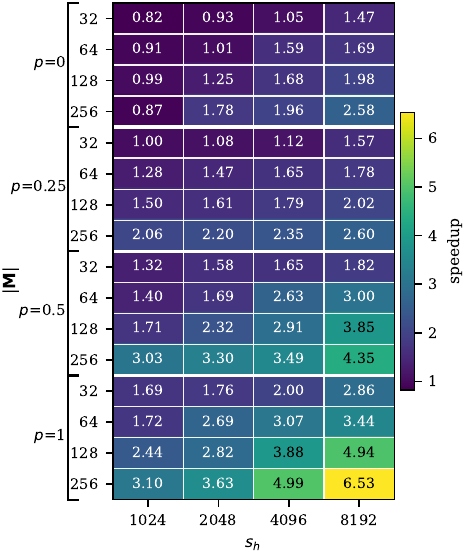}
        \caption*{(a) GPU speedup on host-side static analysis over the sequential baseline.}
    \end{minipage}
    \hfill
    \begin{minipage}[t]{0.51\columnwidth}
        \centering
        \includegraphics[width=\linewidth]{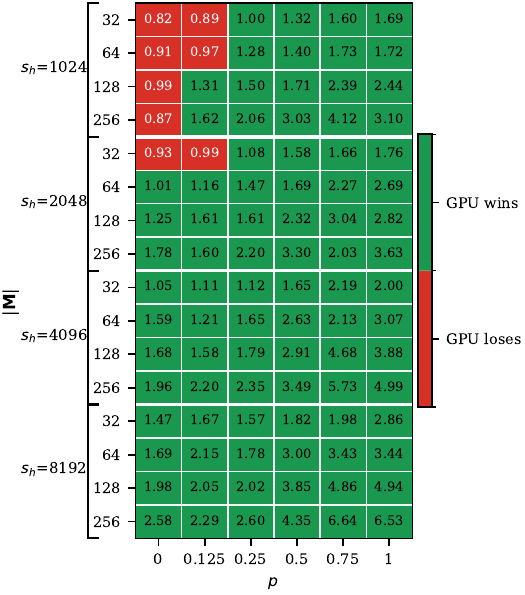}
        \caption*{(b) Break-even view of the GPU acceleration.}
    \end{minipage}
    \caption{GPU acceleration results for host-side static analysis.}
    \label{fig:gpu-speedup-combined}
\end{figure}


\section{Conclusion}\label{sec:conclusion}
We present a static analysis to identify non-contributory measurement outcomes in hybrid quantum-classical programs and drive DGE circuit optimization.
Our pass is orthogonal to existing circuit optimizers, yielding substantial circuit simplification both on its own and in composition with Qiskit, t$\ket{\text{ket}}$, and PyZX.
We further showed that, after lowering exposes sufficient level-wise parallelism, the host-side analysis admits an efficient GPU backend.

\clearpage

\appendix

\counterwithin{algocf}{section}
\counterwithin{theorem}{section}
\counterwithin{corollary}{section}
\counterwithin{lemma}{section}
\counterwithin{equation}{section}
\counterwithin{proposition}{section}

\section{Appendix}\label{sec:appendix}
\subsection{Evaluation on diagnostic workloads}\label{app:diagnostic-and-savings}
In addition to the application-faithful workloads in \cref{subsec:effect}, we also implement three small diagnostic workloads with VQE-style and QPE-style kernels. These workloads are intentionally
small and inspectable: the expected non-contributory measurements can be checked manually, which makes them useful for validating individual transfer rules, measurement-to-qubit mappings, and the end-to-end dead-gate elimination
pipeline. They also serve as lightweight smoke tests for reproducing our
artifact: before running the full application-faithful suite, one can run these
diagnostic instances to check that the host analysis, circuit optimizers,
canonicalization step, and pipeline compositions produce the expected behavior.
These workloads are not intended to establish workload representativeness, and
we therefore exclude them from the aggregate statistics and main effectiveness
claims in \Cref{subsec:effect}. We run the same experimental setup as in
\Cref{subsec:effect} and report the diagnostic results in
\Cref{table:resulting-circuits-metrics-toy,table:relative-optimization-toy}.

\textbf{Results.}
The results on our diagnostic workloads also give supportive evidence. Our host-aware pass \hyperlink{passsettings}{\textbf{h}} is effective on its own and yields orthogonal gains when composed with state-of-the-art optimizers (Qiskit, t$\ket{\text{ket}}$, PyZX).
As \cref{table:resulting-circuits-metrics-toy} shows, applying \hyperlink{passsettings}{\textbf{h}} already simplifies the circuits substantially: on \cref{fig:test-case-toy-example} it prunes all gates; on \cref{fig:test-case-vqe} it reduces the metrics from $(82,26,61)$ to $(12,0,6)$, eliminating all two-qubit gates; on \cref{fig:test-case-qpe} it shrinks these metrics from $(40,16,37)$ to $(31,9,31)$.
\Cref{table:relative-optimization-toy} further shows that inserting \hyperlink{passsettings}{\textbf{h}} around Qiskit/t$\ket{\text{ket}}$/PyZX yields large additional reductions over each optimizer’s standalone output.
Running \hyperlink{passsettings}{\textbf{h}} before the optimizer sometimes exposes more simplifications (e.g., $\hyperlink{passsettings}{\textbf{q}}\!\to\!\hyperlink{passsettings}{\textbf{hq}}$ achieves $(100\%,100\%,100\%)$ extra savings on \cref{fig:test-case-toy-example} and $(83\%,100\%,90\%)$ on \cref{fig:test-case-vqe}), while running \hyperlink{passsettings}{\textbf{h}} afterwards remains quite beneficial.
The only minor adverse effect is for t$\ket{\text{ket}}$ on \cref{fig:test-case-qpe} when applied after \hyperlink{passsettings}{\textbf{h}}$\,$: single-qubit gates increase by 13\%, but two-qubit gates drop by 40\%, a favourable trade-off. 
Overall, these results corroborate: our host-semantics-aware optimization prunes result-invariant gates that host-unaware optimizers cannot, and composing \hyperlink{passsettings}{\textbf{h}} with existing toolchains amplifies simplification beyond their reach.

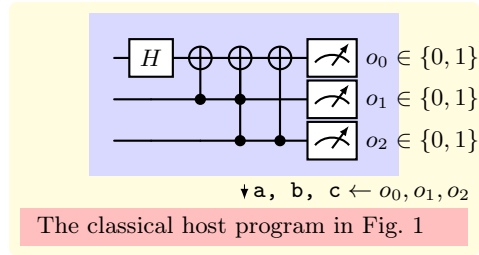
\begin{figure}[htb]
    \centering
\begin{tikzpicture}[>=latex,node distance=1.5em]
  ... 

 \node[](q2)
 {
    \colorbox{blue!15}{
        \begin{minipage}{0.3\textwidth}
\begin{quantikz}[row sep=0.05cm, column sep=0.2cm]
        &\gate{H}&\targ{}   &\targ{}&\targ{}&\meter{}\rstick{$o_0 \in \{0, 1\}$}\\
      &&\ctrl{-1} &\ctrl{-1} &          &\meter{}\rstick{$o_1 \in \{0, 1\}$} \\
      &&&\ctrl{-2} &\ctrl{-2} &\meter{}\rstick{$o_2 \in \{0, 1\}$}
\end{quantikz}
      \end{minipage}
    }
 }; 
 \node[below= 0.2cm of q2](c3)
{
    \colorbox{pink}{
        \begin{minipage}{0.45\textwidth}
        The classical host program in \cref{fig:motivating-example-hybrid-program}
        \end{minipage}
    }
};

 \path[->]
(q2) edge [] node [right] {$\texttt{a, b, c} \gets o_0, o_1, o_2$} (c3);
  \begin{scope}[on background layer]
    \node[
      fit=(q2) (c3),
      inner sep=1pt,
      rounded corners=3pt,
      fill=yellow!15
    ] {};
  \end{scope}
\end{tikzpicture}
    \caption{Diagnostic workload one: a toy hybrid program similar to the \cref{fig:motivating-example-hybrid-program}.}
    \label{fig:test-case-toy-example}
\end{figure}

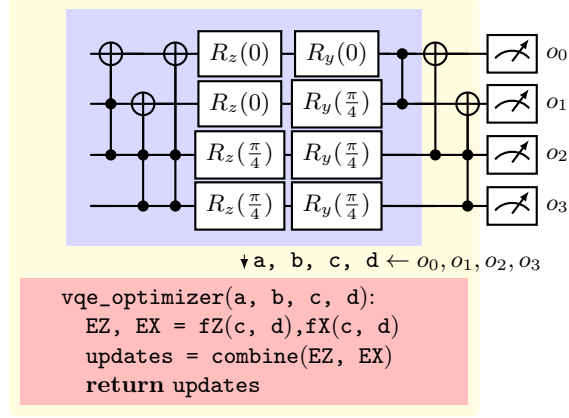
\begin{figure}[htb]
    \centering
\begin{tikzpicture}[>=latex,node distance=1.5em]
  ... 

 \node[](q2)
 {
    \colorbox{blue!15}{
        \begin{minipage}{0.35\textwidth}
\begin{quantikz}[row sep=0.05cm, column sep=0.1cm]
 & \targ{}      & \qw         & \targ{}      & \gate{R_z(0)}      & \gate{R_y(0)}      & \qw        & \ctrl{1} & \qw        & \targ{}     & \qw & \meter{}\rstick{$o_0$}\\
 & \ctrl{-1}    & \targ{}     & \qw          & \gate{R_z(0)}      & \gate{R_y(\frac{\pi}{4})}  &    & \control{} &    & \qw         & \targ{}& \meter{}\rstick{$o_1$}\\
 & \ctrl{-2}    & \ctrl{-1}   & \ctrl{-2}    & \gate{R_z(\frac{\pi}{4})}  & \gate{R_y(\frac{\pi}{4})}  & \qw        & \qw      & \qw        & \ctrl{-2}   & \ctrl{-1}& \meter{}\rstick{$o_2$}\\
 & \qw          & \ctrl{-2}   & \ctrl{-3}    & \gate{R_z(\frac{\pi}{4})}  & \gate{R_y(\frac{\pi}{4})}  & \qw        & \qw      & \qw        & \qw         & \ctrl{-2}& \meter{}\rstick{$o_3$}
\end{quantikz}
      \end{minipage}
    }
 }; 
 \node[below=0.2cm of q2](c3)
{
    \colorbox{pink}{
        \begin{minipage}{0.45\textwidth}
        \SetAlgoLined
            \SetNlSty{textbf}{}{:}
            \begin{algorithmic}
            \State $\texttt{vqe\_optimizer}(\texttt{a, b, c, d})$:
            \State \quad $\texttt{EZ, EX = fZ}(\texttt{c, d})\texttt{,} \texttt{fX}(\texttt{c, d})$       
            \State \quad $\texttt{updates = combine}(\texttt{EZ, EX})$         
            \State \quad \textbf{return} \texttt{updates} 
            \end{algorithmic}
       
        \end{minipage}
    }
};

 \path[->]
(q2) edge [] node [right] {$\texttt{a, b, c, d} \gets o_0, o_1, o_2, o_3$} (c3);
  \begin{scope}[on background layer]
    \node[
      fit=(q2) (c3),
      inner sep=1pt,
      rounded corners=3pt,
      fill=yellow!15
    ] {};
  \end{scope}
\end{tikzpicture}
    \caption{Diagnostic workload two: a VQE-style estimator/update kernel. 
    The Ansatz (blue) uses a fixed 3-qubit entangling backbone followed by a layer of parameterized single-qubit rotations and a CX/CCX correlation-shaping layer to boost expressivity. 
    The classical host (red) aggregates subsets of measurement shots into estimators (e.g., Z- and X-type energy terms), combines them into an objective, and updates the circuit parameters. 
    This pattern matches how VQE, QAOA, and related variational algorithms are implemented in practice.
    }
    \label{fig:test-case-vqe}
\end{figure}

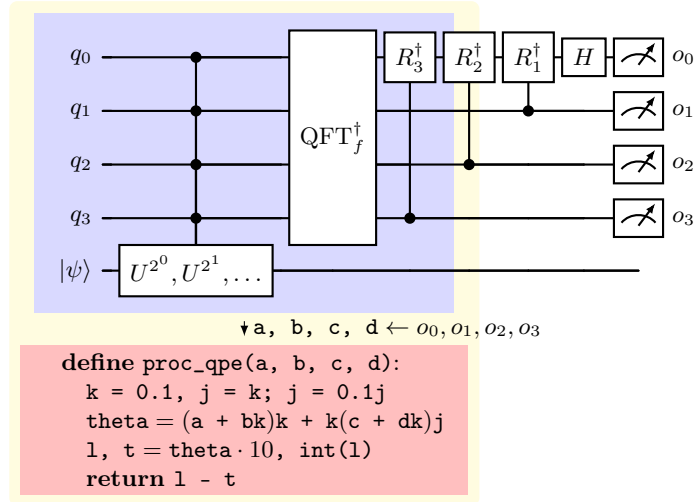
\begin{figure}[htb]
    \centering
\begin{tikzpicture}[>=latex,node distance=1.5em]
  ... 

 \node[](q2)
 {
    \colorbox{blue!15}{
        \begin{minipage}{0.42\textwidth}

\begin{quantikz}[row sep=0cm, column sep=0.11cm]
\lstick{$q_0$} & \qw & \ctrl{4} & \qw
               & \gate[wires=4]{\mathrm{QFT}^{\dagger}_{f}}
               & \gate{R_3^{\dagger}} & \gate{R_2^{\dagger}} & \gate{R_1^{\dagger}} & \gate{H} & \meter{} \rstick{$o_0$} \\
\lstick{$q_1$} & \qw & \ctrl{3} & \qw
               & \ghost{\mathrm{QFT}^{\dagger}_{f}}
               & \qw & \qw & \ctrl{-1} & \qw & \meter{} \rstick{$o_1$} \\
\lstick{$q_2$} & \qw & \ctrl{2} & \qw
               & \ghost{\mathrm{QFT}^{\dagger}_{f}}
               & \qw & \ctrl{-2} & \qw & \qw & \meter{} \rstick{$o_2$} \\
\lstick{$q_3$} & \qw & \ctrl{1} & \qw
               & \ghost{\mathrm{QFT}^{\dagger}_{f}}
               & \ctrl{-3} & \qw & \qw & \qw & \meter{} \rstick{$o_3$} \\
\lstick{$|\psi\rangle$} & \qw & \gate{U^{2^0},U^{2^1},\ldots} & \qw
               & \qw
               & \qw & \qw & \qw & \qw & \qw
\end{quantikz}

      \end{minipage}
    }
 }; 
 \node[below=0.2cm of q2](c3)
{
    \colorbox{pink}{
        \begin{minipage}{0.45\textwidth}
        \SetAlgoLined
            \SetNlSty{textbf}{}{:}
            \begin{algorithmic}
            \State \textbf{define} $\texttt{proc\_qpe(a, b, c, d)}$:
            \State \quad $\texttt{k = 0.1, j = k; j = 0.1j}$
            \State \quad $\texttt{theta} = (\texttt{a + bk})\texttt{k + k}(\texttt{c + dk})\texttt{j}$          
            \State \quad $\texttt{l, t} = \texttt{theta}\cdot10\texttt{, int(l)}$
            \State \quad $\textbf{return } \texttt{l - t}$
            \end{algorithmic}
        \end{minipage}
    }
};

 \path[->]
(q2) edge [] node [right] {$\texttt{a, b, c, d} \gets o_0, o_1, o_2, o_3$} (c3);
  \begin{scope}[on background layer]
    \node[
      fit=(q2) (c3),
      inner sep=1pt,
      rounded corners=3pt,
      fill=yellow!15
    ] {};
  \end{scope}
\end{tikzpicture}
    \caption{Diagnostic workload three: a QPE kernel where the counting register $q_0,\ldots,q_{3}$ controls powers of a single–qubit phase $U(\varphi)=R_z(\varphi)$ applied to the work state $\ket{\psi}$; readout uses $QFT_{f}^{\dagger}$, whose top stage is shown explicitly as $R_3^{\dagger},R_2^{\dagger},R_1^{\dagger},H$ with $R_i^{\dagger}\equiv\operatorname{CP}(-\pi/2^{\,i})$. 
    The host then 
    executes the analyzing program in the red box.}
    \label{fig:test-case-qpe}
\end{figure}


\begin{table}[htb]
\centering
\caption{Gate counts and depth after applying set-up pipelines on the $3$ diagnostic hybrid workloads. Each entry is a triple $(\#\text{1q-gate},\,\#\text{2q-gate},\,\text{depth})$ measured on the resulting circuit. 
Lower values indicate more simplification.}

\begin{tabular}{l|c|c|c|c}
\toprule
\diagbox[width=6.5em]{Pipelines}{Workloads} & \cref{fig:test-case-toy-example} & \cref{fig:test-case-vqe} & \cref{fig:test-case-qpe}  & \\
\midrule
original circuit & $(16, 8, 19)$ & $(82,26,61)$ & $(40       ,16       ,37)$ &   \\
\hyperlink{passsettings}{\textbf{q}} & $(19, 7, 17)$ & $(59       ,25,       54)$  & $(40        ,9       ,28)$ &   \\
\hyperlink{passsettings}{\textbf{t}} & $(24, 7, 26)$ & $( 105       ,25       ,83)$ & $(53       , 10       ,36)$ &   \\
\hyperlink{passsettings}{\textbf{p}} & $(10, 6, 11)$ & $(79       ,32       ,69)$ & $(47,       15       ,24)$ &   \\
\hyperlink{passsettings}{\textbf{h}} & $(0, 0, 10)$ & $(12        ,0        ,6)$ & $(31        ,9       ,31)$  &   \\
\hyperlink{passsettings}{\textbf{qh}} & $(3, 2, 14)$ & $(17        ,6       ,16)$ & $(34        ,5       ,21)$ &   \\
\hyperlink{passsettings}{\textbf{\textbf{h}q}} & $(0,0,0)$  & $(10        ,0        ,5)$ & $( 31,        5       ,22)$ & \\
\hyperlink{passsettings}{\textbf{t\textbf{h}}} & $(3,2,4)$  & $(42        ,6       ,41)$ &  $(41        ,6       ,28)$  & \\
\hyperlink{passsettings}{\textbf{\textbf{h}t}} & $(0,0,0)$  & $( 12        ,0        ,6)$ & $(60,        6       ,39)$  & \\
\hyperlink{passsettings}{\textbf{p\textbf{h}}} & $(3,2,4)$ & $(18        ,2       ,12)$ & $(33,        8       ,22)$ & \\
\hyperlink{passsettings}{\textbf{\textbf{h}p}} & $(0,0,0)$ & $(18        ,0        ,9)$ & $(32        ,8       ,21)$ & \\
\bottomrule

\end{tabular}
\label{table:resulting-circuits-metrics-toy}
\end{table}

\begin{table}[htb]
\centering
\caption{
Percentage additional savings on diagnostic workloads by integrating our host-aware pass \textbf{h} with standard circuit optimizers.
For the row $X\! \rightarrow\! YZ$, each entry $(\Delta\text{1q-gate},\,\Delta\text{2q-gate},\,\Delta\text{depth})$ is 
$100\times\frac{\text{metric after }X-\text{metric after }YZ}{\text{metric after }X}$ 
with the metric being \#1-qubit gates, \#2-qubit gates, and depth.
Positive values indicate further reduction; negatives indicate regressions.
}

\begin{tabular}[htb]{l|c|c|c|c}
\toprule
\diagbox[width=6.5em]{Compare}{Workloads} & \cref{fig:test-case-toy-example} & \cref{fig:test-case-vqe} & \cref{fig:test-case-qpe}  & \\
\midrule
$\hyperlink{passsettings}{\textbf{q}}\rightarrow\hyperlink{passsettings}{\textbf{qh}}$ & $( 84       ,71         ,76)$ & $(71       ,76         ,70)$  & $( 15       ,44         ,25)$ &   \\
$\hyperlink{passsettings}{\textbf{q}}\rightarrow\hyperlink{passsettings}{\textbf{hq}}$ & $(100      ,100        ,100)$ & $(83,      100         ,90)$ & $(22       ,44         ,21)$ &   \\
$\hyperlink{passsettings}{\textbf{t}}\rightarrow\hyperlink{passsettings}{\textbf{th}}$ & $( 87       ,71         ,84)$ & $(60       ,76         ,50)$ & $(22       ,40        ,22)$ &   \\
$\hyperlink{passsettings}{\textbf{t}}\rightarrow\hyperlink{passsettings}{\textbf{ht}}$ & $(100      ,100       ,100)$ & $(88      ,100         ,92)$ & $(-13       ,40         ,-8)$  &   \\
$\hyperlink{passsettings}{\textbf{p}}\rightarrow\hyperlink{passsettings}{\textbf{ph}}$ & $(70      ,66        ,63)$ & $(77       ,93         ,82)$ & $( 29       ,46          ,8)$ &   \\
$\hyperlink{passsettings}{\textbf{p}}\rightarrow\hyperlink{passsettings}{\textbf{hp}}$ & $(100,      100,        100)$  & $(77      ,100         ,87)$ & $(31       ,46         ,12)$ & \\
\bottomrule

\end{tabular}
\label{table:relative-optimization-toy}
\end{table}
\clearpage
\subsection{Pseudo-code implementations}\label{app:helper-proc}
In \cref{alg:init-abs-state}, \textsc{InitAbsState} creates the initial abstract environment and
marks measurement-bound variables as binary inputs. In \cref{alg:analyze-block}, \textsc{AnalyzeBlock}
performs one forward abstract execution of the host block, using the expression transfer rules in \cref{subsec:abs-eval} and the statement effects in \cref{subsec:abs-effect}.
In \cref{alg:find-nc},
\textsc{find\_nc} further uses the outcome of \cref{alg:init-abs-state,alg:analyze-block} to compute a set of non-contributory measurement outcomes.

\cref{alg:syntactic-live,alg:live-block} describe a standard backward live-variable analysis over the host program, augmented only with conservative control dependence through branch guards.  The operator $\mathrm{Vars}(e)$ returns variables syntactically occurring in $e$; it does not perform algebraic simplification.  Thus, for example, $\mathrm{Vars}(0\cdot\texttt{m})=\{\texttt{m}\}$ and equivalent branch returns are not identified as equal. 

In \cref{alg:host-aware-simpl}, $Q$ denotes the set of qubits. The call $C_{\text{opt}}.\textbf{frontier}()$ returns the current output-side frontier: the operators adjacent to terminal measurements. The procedure \textsc{CloseUnderCones}$(C_{\text{opt}},\mathbf{Q}_{\text{dead}})$ propagates deadness backward through the remaining circuit: whenever all output-side wire segments of a operator lead only to qubits in $\mathbf{Q}_{\text{dead}}$, the input-side wire segments are also marked dead.

\begin{algorithm}[htb]
\caption{Host-aware circuit simplification}
\label{alg:host-aware-simpl}
\KwData{$C \in \textsc{circuits}$, host code $H$, measurement map $\mu$}
\KwResult{$C_{\text{opt}}$}
$C_{\text{opt}} \gets C$,\quad $terminate \gets \texttt{false}$\;
$\mathbf{M}_{\text{nc}}\! \gets\! \textsc{find\_nc}(H,\mu)$, 
$\mathbf{Q}_{\text{dead}}\! \gets\! \{ q \mid\! \mu(q)\in \mathbf{M}_{\text{nc}} \}$\;
\While{$terminate \neq \texttt{true}\ \land\
       \emptyset \neq \mathcal{F}_C \gets C_{\text{opt}}.\textbf{frontier}()$}{
  $terminate \gets \texttt{true}$\;
  \For{$g \in \mathcal{F}_C$}{
    \If{\textsc{IsDeadGate}\,($g,\mathbf{Q}_{\text{dead}}$)}{
      $C_{\text{opt}} \gets C_{\text{opt}} - g$,\quad
      $terminate \gets \texttt{false}$\;
    }
  }
  $\mathbf{Q}_{\text{dead}} \gets \textsc{CloseUnderCones}(C_{\text{opt}},\mathbf{Q}_{\text{dead}})$\;
}
\Return $C_{\text{opt}}$\;
\end{algorithm}

\begin{algorithm}[htb]
\caption{\textsc{find\_nc}: detecting non-contributory measurements}
\label{alg:find-nc}
\KwData{host code $H$, measurement map $\mu:Q\rightharpoonup\mathcal{V}$}
\KwResult{non-contributory measurement-bound variables $\mathbf{M}_{\mathrm{nc}}$}

$\mathbf{M}\gets \{\mu(q)\mid q\in\mathrm{dom}(\mu)\}$\;
$(\rho_0,\mathbf{Dom},\mathbf{B})\gets
\textsc{InitAbsState}(H,\mathbf{M})$ \tcp*{See \cref{alg:init-abs-state}} 

$(\rho^\star,C_{\texttt{ctrl}}^\star)
\gets
\textsc{AnalyzeBlock}(H,\rho_0,\varnothing,\mathbf{M})$
 \tcp*{See \cref{alg:analyze-block}}

$\mathbf{Obs}^\star
\gets
\bigl(
\mathbf{Sym}(\rho^\star(\texttt{return}))
\cap
\{\texttt{m}@0\mid \texttt{m}\in\mathbf{M}\}
\bigr)
\cup
\{\texttt{m}@0\mid \texttt{m}\in C_{\texttt{ctrl}}^\star\}$\;

$\mathbf{M}_{\mathrm{nc}}
\gets
\{\texttt{m}\in\mathbf{M}
\mid
\texttt{m}@0\notin\mathbf{Obs}^\star\}$\;

\Return $\mathbf{M}_{\mathrm{nc}}$\;
\end{algorithm}

\begin{algorithm}[htb]
\caption{\textsc{InitAbsState}: initializing the abstract state}
\label{alg:init-abs-state}
\KwData{host code $H$, measurement-bound variables $\mathbf{M}$}
\KwResult{initial abstract state $(\rho_0,\mathbf{Dom},\mathbf{B})$}

$\mathcal{V}\gets\mathrm{Vars}(H)$,\quad
$\Sigma_0\gets\{\texttt{x}@0\mid \texttt{x}\in\mathcal{V}\}$\;

\ForEach{$\texttt{x}\in\mathcal{V}$}{
    $\rho_0(\texttt{x})\gets\langle\texttt{x}@0\rangle$\;
    $\mathbf{Dom}(\texttt{x}@0)\gets\textsc{InitDom}(\texttt{x},H)$, 
    $\mathbf{B}(\texttt{x}@0)\gets\textsc{InitBound}(\texttt{x},H)$\;
}

$\rho_0(\texttt{return})\gets\bot$\;

\ForEach{$\texttt{m}\in\mathbf{M}$}{
    $\mathbf{Dom}(\texttt{m}@0)\gets\textbf{bin}$,\quad
    $\mathbf{B}(\texttt{m}@0)\gets[0,1]$\;
}

\Return $(\rho_0,\mathbf{Dom},\mathbf{B})$\;
\end{algorithm}

\begin{algorithm}[htb]
\caption{\textsc{AnalyzeBlock}: abstract execution of a host block}
\label{alg:analyze-block}
\KwData{statement block $S_b$, abstract state $\rho$, control set $C_{\texttt{ctrl}}$,
measurement-bound variables $\mathbf{M}$}
\KwResult{updated abstract state $(\rho,C_{\texttt{ctrl}})$}

\ForEach{$S \in S_b$}{
    \uIf{$S \equiv \texttt{x := } e$}{
        $\rho \gets \rho \oplus \{\texttt{x}\mapsto [\![e]\!]^\sharp_\rho\}$\;
    }
    \uElseIf{$S;S' \equiv
    \texttt{if (cond) then }S_t\texttt{ else }S_e\texttt{; return }e$ $\wedge$
    $S_t,S_e$ contain no \texttt{return}}{
        $\rho_{\mathrm{in}}\gets\rho$\;
        $(\rho_t,C_{\texttt{ctrl}}^t)
        \gets
        \textsc{AnalyzeBlock}(S_t,\rho_{\mathrm{in}},C_{\texttt{ctrl}},\mathbf{M})$\;
        $(\rho_e,C_{\texttt{ctrl}}^e)
        \gets
        \textsc{AnalyzeBlock}(S_e,\rho_{\mathrm{in}},C_{\texttt{ctrl}},\mathbf{M})$\;
        $R_t\gets[\![e]\!]^\sharp_{\rho_t}$,\quad
        $R_e\gets[\![e]\!]^\sharp_{\rho_e}$\;
        $\rho\gets
        (\rho_t\sqcup_{\mathrm{env}}\rho_e)
        \oplus\{\texttt{return}\mapsto R_t\sqcup_{\mathrm A}R_e\}$\;
        $C_{\texttt{ctrl}}\gets C_{\texttt{ctrl}}^t\cup C_{\texttt{ctrl}}^e$\;
        \If{$\neg\mathsf{Eq}_{\mathrm A}(R_t,R_e)$}{
            $C_{\texttt{ctrl}}\gets C_{\texttt{ctrl}}
            \cup
            \{\texttt{m}\in\mathbf{M}
            \mid
            \texttt{m}@0\in\mathbf{Sym}([\![\texttt{cond}]\!]^\sharp_{\rho_{\mathrm{in}}})\}$\;
        }
        \Return $(\rho,C_{\texttt{ctrl}})$\;
    }
    \uElseIf{$S \equiv \texttt{return } e$}{
        $\rho \gets \rho \oplus \{\texttt{return}\mapsto [\![e]\!]^\sharp_\rho\}$\;
        \Return $(\rho,C_{\texttt{ctrl}})$\;
    }
    \uElseIf{$S \equiv \texttt{if (cond) then }S_t\texttt{ else }S_e$}{
        $\rho_{\mathrm{in}}\gets\rho$\;
        $(\rho_t,C_{\texttt{ctrl}}^t)
        \gets
        \textsc{AnalyzeBlock}(S_t,\rho_{\mathrm{in}},C_{\texttt{ctrl}},\mathbf{M})$\;
        $(\rho_e,C_{\texttt{ctrl}}^e)
        \gets
        \textsc{AnalyzeBlock}(S_e,\rho_{\mathrm{in}},C_{\texttt{ctrl}},\mathbf{M})$\;

        $\rho\gets \rho_t\sqcup_{\mathrm{env}}\rho_e$\;
        $C_{\texttt{ctrl}}\gets C_{\texttt{ctrl}}^t\cup C_{\texttt{ctrl}}^e$\;

        \If{$\exists\texttt{x}\in\mathcal{V}\cup\{\texttt{return}\}:
        \neg\mathsf{Eq}_{\mathrm A}(\rho_t(\texttt{x}),\rho_e(\texttt{x}))$}{
            $C_{\texttt{ctrl}}\gets C_{\texttt{ctrl}}
            \cup
            \{\texttt{m}\in\mathbf{M}
            \mid
            \texttt{m}@0\in\mathbf{Sym}([\![\texttt{cond}]\!]^\sharp_{\rho_{\mathrm{in}}})\}$\;
        }
    }
}
\Return $(\rho,C_{\texttt{ctrl}})$\;
\end{algorithm}

\begin{algorithm}[htb]
\caption{\textsc{SyntacticLive}: syntactic host-liveness baseline}
\label{alg:syntactic-live}
\KwData{host block $H$, measurement-bound variables $\mathbf{M}$, measurement map $\mu:Q\rightharpoonup\mathcal{V}$}
\KwResult{syntactically non-contributory measurement variables $\mathbf{M}_{\mathrm{syn}}$}

$\mathbf{L}\gets\textsc{LiveBlock}(H,\varnothing)$ \tcp{Backward analysis in \cref{alg:live-block}}
$\mathbf{M}_{\mathrm{syn}}\gets\{\texttt{m}\in\mathbf{M}\mid \texttt{m}\notin\mathbf{L}\}$\;
\Return $\mathbf{M}_{\mathrm{syn}}$\;
\end{algorithm}

\begin{algorithm}[htb]
\caption{\textsc{LiveBlock}: backward syntactic transfer}
\label{alg:live-block}
\KwData{statement block $S_b$, live set $\mathbf{L}$}
\KwResult{live variables before $S_b$}

\ForEach{$S\in\mathrm{reverse}(S_b)$}{
    \uIf{$S\equiv\texttt{return }e$}{
        $\mathbf{L}\gets\mathrm{Vars}(e)$\;
    }
    \uElseIf{$S\equiv\texttt{x := }e$}{
        \If{$\texttt{x}\in\mathbf{L}$}{
            $\mathbf{L}\gets(\mathbf{L}\setminus\{\texttt{x}\})\cup\mathrm{Vars}(e)$\;
        }
    }
    \uElseIf{$S\equiv\texttt{if (cond) then }S_t\texttt{ else }S_e$}{
        $\mathbf{L}_t\gets\textsc{LiveBlock}(S_t,\mathbf{L})$,
        $\mathbf{L}_e\gets\textsc{LiveBlock}(S_e,\mathbf{L})$\;
        $\mathbf{L}\gets\mathbf{L}_t\cup\mathbf{L}_e$\;
        \If{$\textsc{MayAffectLive}(S_t,S_e,\mathbf{L})$}{
            $\mathbf{L}\gets\mathbf{L}\cup\mathrm{Vars}(\texttt{cond})$\;
        }
    }
}
\Return $\mathbf{L}$\;
\end{algorithm}
\clearpage
\subsection{Application-faithful workload descriptions}\label{app:workloads}
\paragraph{Quantum chemistry workloads.}
These workloads model VQE-style molecular-estimator post-processing.  
The quantum kernels prepare parameterized Ansatz states and measure Pauli-style observable groups.  
The host programs map measured bits to Pauli
eigenvalue variables, aggregate weighted Pauli products, and return an energy or selected molecular-property estimate.  
Such host-side Hamiltonian averaging and measurement aggregation are standard in VQE and near-term quantum chemistry workflows
\cite{Peruzzo2014VQE,McClean2016VQE,McArdle2020QChem}.

\begin{description}
\item[\textbf{C1-pauli}.]
This workload represents a VQE energy-estimation step after Hamiltonian
coefficient thresholding.  The quantum kernel is an irregular RealAmplitudes
VQE-style ansatz followed by measurements for Pauli-estimator groups.  The host
program converts each measurement bit into a Pauli eigenvalue and computes a
weighted sum of Pauli products.  The upstream chemistry artifact contains a
screened low-magnitude Hamiltonian tail term; therefore, the measurements used
only by that screened term no longer affect the returned energy.  This models
the common situation in which a chemistry workflow discards small Hamiltonian
coefficients or measurement terms to reduce estimator cost
\cite{Peruzzo2014VQE,McClean2016VQE,Verteletskyi2020Measurement}.

\item[\textbf{C2-symmetry}.]
This workload models an estimator that uses known symmetry or tapering-sector
information.  The quantum kernel is an EfficientSU2/two-local VQE-style ansatz
with nonuniform measurement groups.  The host program first applies upstream
constant replacements for measurements whose Pauli eigenvalues are fixed by a
symmetry-sector certificate, and then evaluates the remaining energy terms.
The quantum and host sides cooperate as follows: the quantum circuit still
produces measurement variables, but the application-level host computation
knows that some variables are fixed by the physical sector and therefore
overwrites them with constants.  This reflects symmetry-based qubit reduction
and tapering techniques used in fermionic Hamiltonian simulation
\cite{Bravyi2017Tapering,McArdle2020QChem}.

\item[\textbf{C3-active-space}.]
This workload models active-space chemistry with frozen-orbital constants.
The quantum kernel is an active-space VQE-style ansatz with an irregular
tail corresponding to frozen or inactive orbital lines.  The host program
aggregates active-space Pauli terms but replaces selected frozen-orbital
measurements by constants supplied by the upstream active-space artifact.
Thus, the host return value depends only on the active-space portion of the
measurement record.  This reflects standard chemistry preprocessing in which
inactive orbitals are treated classically while only the active space is
handled by the quantum computation
\cite{McArdle2020QChem,QiskitNatureFreezeCore}.

\item[\textbf{C4-screened-obs}.]
This workload represents grouped observable evaluation after screening
low-weight Pauli blocks.  The quantum kernel contains heterogeneous grouped
VQE observable blocks.  The host program forms a molecular estimator by
summing the retained observable groups.  The upstream artifact records that
one low-weight observable block was removed by application-level screening,
so measurements used only by that block become semantically irrelevant.  This
captures the common VQE practice of grouping and filtering observable terms
to reduce measurement overhead
\cite{Verteletskyi2020Measurement,McClean2016VQE}.

\item[\textbf{C5-property-report}.]
This workload models a shared quantum estimator used for multiple molecular
properties, where the application reports only a selected property subset.
The quantum kernel measures heterogeneous observable blocks that could support
several property reports.  The host program aggregates only the observable
blocks requested by the current report.  A block corresponding to an unreported
property is still present in the shared quantum kernel, but its measurements
do not contribute to the host return value.  This is application-faithful for
chemistry workflows that reuse estimator templates across related observables
while reporting only selected quantities
\cite{McArdle2020QChem,Verteletskyi2020Measurement}.

\item[\textbf{C6-lowrank-tail}.]
This workload models a low-rank chemistry estimator with a screened tail
factor.  The quantum kernel follows an irregular VQE-style estimator template
with nonuniform groups corresponding to retained and tail factors.  The host
program computes the estimator by summing the retained low-rank terms.  The
upstream artifact records that a small low-rank tail factor was removed, so
measurements used only by that tail factor are dead with respect to the final
energy estimate.  This reflects low-rank Hamiltonian decompositions and
truncation strategies used to reduce the resource cost of electronic-structure
simulation
\cite{Motta2021LowRank,McArdle2020QChem}.
\end{description}

\paragraph{Quantum optimization workloads.}
These workloads model QAOA/QUBO-style workflows.  
The quantum kernels prepare
candidate binary assignments using irregular QAOA-style circuits.  The host
program interprets measurement bits as decision variables and evaluates an
application objective, typically a weighted QUBO expression.  The upstream
artifacts model presolve, component selection, fixed variables, disabled
constraints, or scenario filtering.  QAOA and QUBO formulations are standard
near-term optimization models
\cite{Farhi2014QAOA,Glover2019QUBO}.

\begin{description}
\item[\textbf{O1-inactive-var}.]
This workload models a QUBO objective after preprocessing removes all
objective terms incident to one decision variable.  The quantum kernel is a
QAOA/QUBO circuit that still measures all candidate variables.  The host
program evaluates the preprocessed objective by summing the nonzero weighted
edges.  The inactive variable may still be measured by the quantum kernel, but
because all of its incident objective coefficients are zero in the upstream
presolve artifact, it cannot affect the returned score.  This reflects optimization pipelines in which classical presolve removes inactive variables before objective evaluation
\cite{Farhi2014QAOA,Glover2019QUBO}.

\item[\textbf{O2-component-select}.]
This workload models a disconnected QUBO instance where the application reports
only a selected connected component.  The quantum kernel prepares and measures
variables for both the selected component and an omitted component.  The host
program evaluates only the selected component's objective edges.  Measurements
belonging solely to the unreported disconnected component are irrelevant to
the returned objective value.  This captures the application pattern in which
a decomposed optimization model is solved or reported component-wise
\cite{Farhi2014QAOA,Glover2019QUBO}.

\item[\textbf{O3-penalty-filter}.]
This workload represents a QUBO objective with disabled penalty or constraint
blocks.  The quantum kernel measures variables that would support both the
main objective and several penalty terms.  The host program evaluates the
objective using the active upstream penalty configuration.  Penalty blocks
with zero weights are excluded from the returned score, so variables used only
inside disabled blocks do not influence the host return value.  This models
application workflows in which constraints, soft penalties, or regularizers
are conditionally enabled or disabled by a model export
\cite{Glover2019QUBO,Farhi2014QAOA}.

\item[\textbf{O4-fixed-vars}.]
This workload models QUBO presolve with fixed decision variables.  The quantum
kernel measures a full candidate assignment, but the upstream artifact records
that selected variables have been fixed by classical preprocessing.  The host
program first overwrites those variables with constants and then evaluates the
QUBO objective.  A measured variable that is fixed before the objective is
computed no longer contributes semantically to the returned score.  This
reflects standard optimization presolve, where variables can be fixed before
the remaining subproblem is evaluated
\cite{Glover2019QUBO,Farhi2014QAOA}.

\item[\textbf{O5-multiobj}.]
This workload models a multi-objective QAOA/QUBO report in which the
application queries only one objective component.  The quantum kernel measures
variables for heterogeneous objective components.  The host program returns
the selected objective score and omits the unqueried component from the
current report.  Measurements used only by the omitted component are therefore
dead with respect to the current host return value.  This is faithful to
optimization workflows that share a quantum candidate generator across several
objectives but report only a selected objective for a particular query
\cite{Glover2019QUBO,Farhi2014QAOA}.

\item[\textbf{O6-robust-scen}.]
This workload models robust or scenario-weighted optimization.  The quantum
kernel prepares a QAOA-style assignment over variables that participate in
several scenario components.  The host program evaluates a robust objective
using only scenarios retained by the upstream scenario filter.  Measurements
associated solely with filtered scenarios do not affect the returned robust
score.  This captures application workflows where a QUBO-style objective is
evaluated under a subset of active scenarios or risk cases
\cite{Glover2019QUBO,Farhi2014QAOA}.
\end{description}

\paragraph{Quantum machine learning workloads.}
These workloads model quantum machine learning pipelines.  
The quantum kernels compute feature channels, support-vector/kernel estimates, ensemble features, or class-head features.  
The host program aggregates these measured features using a trained classical readout, such as a sparse linear head, an SVM-style decision function, an ensemble weight vector, or a selected class
query.  
Quantum feature maps and quantum kernel methods are standard QML patterns, and support-vector models naturally produce sparse decision functions through zero dual coefficients
\cite{Havlicek2019QuantumFeatures,Schuld2019FeatureHilbert,Cortes1995SVM}.

\begin{description}
\item[\textbf{M1-sparse-readout}.]
This workload models a QNN feature extractor followed by a sparse linear
readout head.  The quantum kernel consists of batched, nonuniform QNN
feature-channel circuits.  The host program converts measurements into channel
features and computes a linear logit using trained readout weights.  The
upstream trained-model artifact contains a zero weight for one feature
channel, as would arise from sparsity-inducing training or feature selection.
That channel is still produced by the quantum kernel, but it does not affect
the returned logit
\cite{Havlicek2019QuantumFeatures,Schuld2019FeatureHilbert,Tibshirani1996Lasso}.

\item[\textbf{M2-kernel-svm}.]
This workload models a quantum-kernel/SVM decision function.  
The quantum kernel computes heterogeneous quantum feature or kernel-estimator blocks associated with candidate support vectors.  
The host program evaluates an SVM-style decision score by weighting these quantum features with exported
dual coefficients.  Non-support vectors have zero dual coefficients in the upstream model export, so measurements used only for those entries are
irrelevant to the final decision score.  
This follows the standard SVM structure, where the decision function is sparse in support vectors
\cite{Havlicek2019QuantumFeatures,Schuld2019FeatureHilbert,Cortes1995SVM}.

\item[\textbf{M3-ensemble}.]
This workload models an ensemble of QNN submodels.  The quantum kernel
contains batched feature-extraction blocks for several submodels.  The host
program computes an ensemble score by summing submodel outputs with exported
ensemble weights.  The upstream artifact assigns zero weights to pruned
submodels, so their corresponding measured features do not affect the final
ensemble score.  This reflects a common inference pattern in which an ensemble
or model-selection stage keeps only a subset of candidate predictors
\cite{Havlicek2019QuantumFeatures,Schuld2019FeatureHilbert,Tibshirani1996Lasso}.

\item[\textbf{M4-sensor-drop}.]
This workload models sensor or feature-channel selection before QML inference.
The quantum kernel computes nonuniform feature channels, each representing a
different input sensor or feature group.  The host program aggregates only the
channels included in the upstream sensor manifest.  A dropped sensor channel
may still be present in the shared quantum kernel template, but its feature is
not consumed by the host return value.  This captures application pipelines in
which low-quality or unavailable feature channels are excluded at inference
time
\cite{Havlicek2019QuantumFeatures,Schuld2019FeatureHilbert}.

\item[\textbf{M5-multiclass}.]
This workload models a multiclass QNN with several class heads, where the
application queries only a subset of class scores.  The quantum kernel
computes heterogeneous feature blocks that can feed multiple class heads.  The
host program evaluates only the requested class score or subset of scores.
Measurements used exclusively by unqueried class heads are therefore
semantically irrelevant to the current return value.  This models selective
inference in multiclass QML pipelines
\cite{Havlicek2019QuantumFeatures,Schuld2019FeatureHilbert}.

\item[\textbf{M6-calib-filter}.]
This workload models calibrated QNN inference with feature-channel quality
filtering.  The quantum kernel produces several feature channels, including an
auxiliary channel marked by the upstream calibration artifact as low quality.
The host program computes the calibrated readout using only channels that pass
the exported calibration filter.  Measurements used only by the excluded
channel do not contribute to the returned logit.  This is faithful to QML
pipelines where calibration, data quality, or confidence filters determine
which quantum features are consumed by the classical readout
\cite{Havlicek2019QuantumFeatures,Schuld2019FeatureHilbert}.
\end{description}

\paragraph{Quantum finance workloads.}
These workloads model amplitude-estimation-style finance kernels followed by
classical payoff, exposure, or risk-reporting logic.  The quantum kernels
encode price, payoff, scenario, or risk-factor information into measured
registers.  The host program decodes those registers and computes a portfolio
value, payoff bin, VaR-style threshold, or sensitivity report.  This matches
the standard quantum-finance pattern in which quantum amplitude estimation or
state-preparation circuits are followed by classical post-processing for
pricing and risk metrics
\cite{Woerner2019Risk,Stamatopoulos2020Option,Montanaro2015MonteCarlo}.

\begin{description}
\item[\textbf{F1-zero-exposure}.]
This workload models a batched portfolio payoff computation with one
zero-exposure asset or scenario.  The quantum kernel contains AE-style
scenario/payoff blocks, and the host program decodes each measured block into
a payoff value.  It then multiplies each payoff by an exposure exported by the
portfolio artifact and returns the portfolio sum.  A scenario with zero net
exposure is still measured by the quantum kernel but does not affect the
portfolio value.  This reflects portfolio netting or exposure filtering in
financial risk and pricing workflows
\cite{Woerner2019Risk,Stamatopoulos2020Option}.

\item[\textbf{F2-coarse-binning}.]
This workload models coarse-binned option payoff evaluation.  The quantum
kernel produces a price register.  The host program decodes the register with
application-level bit weights, applies integer/floor-style binning, and
returns a payoff derived from the coarse bin.  Low-order residual bits that
cannot change the selected bin are irrelevant to the returned payoff.  This
captures payoff policies where only coarse price buckets matter for the
reported financial quantity
\cite{Stamatopoulos2020Option,Montanaro2015MonteCarlo}.

\item[\textbf{F3-scenario-filter}.]
This workload models scenario-weighted risk aggregation.  The quantum kernel
contains batched AE-style risk-scenario blocks.  The host program decodes each
scenario payoff and multiplies it by an eligibility or scenario weight from
the upstream risk artifact.  Scenarios filtered by zero eligibility weights do
not contribute to the returned risk quantity.  This matches risk-analysis
workflows where only selected scenarios are included in a particular report
\cite{Woerner2019Risk,Egger2021CreditRisk}.

\item[\textbf{F4-hedged-leg}.]
This workload models a basket or hedge payoff computation where one leg is
hedged out.  The quantum kernel prepares and measures payoff registers for
multiple basket or hedge legs.  The host decodes each leg and combines
them using net exposure coefficients.  The upstream portfolio artifact records
that one hedge leg has zero net coefficient after preprocessing, so
measurements used only for that leg are dead.  This reflects portfolio
netting and hedging logic in option-pricing and risk workflows
\cite{Stamatopoulos2020Option,Woerner2019Risk}.

\item[\textbf{F5-var-bucket}.]
This workload models a VaR-style threshold report after coarse risk bucketing.
The quantum kernel produces a risk-score register.  The host decodes
the register, applies coarse bucketization, and returns an indicator of
whether the bucket crosses a risk threshold.  Low-order bits below the bucket
resolution cannot affect the threshold outcome, and therefore do not
contribute to the host return value.  This is faithful to risk analysis
workflows that report thresholded risk quantities such as Value at Risk
\cite{Woerner2019Risk,Egger2021CreditRisk}.

\item[\textbf{F6-greek-filter}.]
This workload models a sensitivity or Greek report with filtered risk factors.
The quantum kernel computes batched AE-style estimators for several factor
groups.  The host program decodes each factor output and aggregates only the
risk factors retained by the upstream sensitivity report.  Factors that are
hedged, inactive, or excluded from the report have zero host weights and their
measurements do not affect the returned sensitivity.  This follows the
finance pattern of combining quantum-estimated payoff or risk quantities with
classical exposure and reporting logic
\cite{Stamatopoulos2020Option,Woerner2019Risk}.
\end{description}

\paragraph{Control workloads.}
The control workloads are designed to validate the host-side analysis rather
than to measure effectiveness.  They are excluded from aggregate effectiveness
statistics.  Their purpose is to check that the analysis finds the expected
dead set in a positive case and, more importantly, does not prune live
measurements in all-live or zero-weight-trap cases.

\begin{description}
\item[\textbf{CTRL1-all-live-qml}.]
This negative control reuses the QML feature-channel structure of
\textbf{M1}, but the upstream trained-model artifact assigns nonzero readout
weights to every feature channel.  The quantum kernel therefore produces
features that are all consumed by the host readout.  The expected dead set is
empty, and the host-aware pass should not reduce the circuit.

\item[\textbf{CTRL2-all-live-opt}.]
This negative control reuses a QAOA/QUBO-style kernel, but every measured
decision variable participates in at least one nonzero objective edge.  The
host program evaluates a fully live QUBO objective, so each measurement can
affect the returned score.  The expected dead set is empty, and any pruning
would indicate a false positive.

\item[\textbf{CTRL3-derived-screened-chem}.]
This positive control reuses a VQE-style Pauli-estimator kernel and
applies the same upstream screening logic as the chemistry workloads.  The
upstream artifact removes a small tail observable, and the expected dead set is
known in advance.  The control checks that the host-side analysis can recover
the screened measurements exactly.

\item[\textbf{CTRL4-zero-weight-live-escape}.]
This negative trap control reuses the QML feature-channel structure of
\textbf{M1}.  The upstream artifact contains a zero primary readout weight for
one feature, but the same feature is also used by an auxiliary live calibration
term.  The correct result is therefore an empty dead set.  This control
demonstrates that the analysis is semantic: it does not prune a measurement
merely because the measurement appears in one zero-weighted term.
\end{description}
\subsection{Soundness of the host-side static analysis}\label{app:soundness-proofs}
\Cref{subsec:host-side-analysis-soundness} summarizes the concrete domain, the abstract-state order, the description relation \(\Delta\), and the
concretization \(\gamma\). We now give the detailed proofs.

\begin{lemma}[Expression-transfer soundness]
\label{lem:expr-transfer-soundness}
$\forall s=(\nu,\sigma,\delta)\Delta(\rho,C_{\texttt{ctrl}}):$
if the concrete evaluation of expression \(e\) under \(s\) yields value
\(v\) with exact measurement-dependency set \(S_D\), and the abstract evaluation
yields $A=[\![e]\!]^{\sharp}_{\rho}$,
then
$(v,S_D)\Delta(A,C_{\texttt{ctrl}})$.
\end{lemma}

\begin{proof}
By structural induction on \(e\), following the transfer rules in
\cref{subsec:abs-eval}.

For constants, \(S_D=\varnothing\), and the abstract value is a constant polynomial, so the claim is immediate. 
For variables, the result follows
directly from \(s\Delta(\rho,C_{\texttt{ctrl}})\).

For binary arithmetic and unary minus, when the operands are in polynomial form, the abstract transfer performs the same symbolic arithmetic and normalizes the polynomial. Hence the concrete value is exactly the valuation of the resulting polynomial, and the measurement support is included in
\(\mathbf{MSym}_{\mathbf{M}}(A)\). 
When exact polynomial reasoning is not
available, the abstract transfer falls back to dependence form and unions the operand supports, which conservatively contains the concrete dependency set.

For \texttt{int}, each exact rule is guarded by a side condition that makes the abstract result equal to the concrete integer result: the argument is already an
integer polynomial, the interval evaluates to a constant integer, or the bounded residual cannot affect the integer part. If none of these conditions holds, the
fallback dependence-form rule conservatively records the argument support.

For \texttt{random}$(L,U)$, the concrete semantics chooses a value in \([L,U]\), while the abstract semantics introduces a fresh non-measurement
symbol with the same bound. This introduces no measurement dependency. For pure function calls, the abstract transfer conservatively unions the dependencies of the arguments, which over-approximates the measurement dependencies of the concrete call result.
\qed
\end{proof}

\begin{lemma}[Edge-effect soundness]
\label{lem:edge-effect-soundness}
Assume
$s\Delta(\rho,C_{\texttt{ctrl}})$.
For every statement or block \(S_b\), if the concrete effect is
$s \xrightarrow{\ S_b\ } s'$
and the abstract effect is
$(\rho,C_{\texttt{ctrl}})
\xRightarrow{\ S_b\ }
(\rho',C_{\texttt{ctrl}}')$,
then
$s'\Delta(\rho',C_{\texttt{ctrl}}')$.

\end{lemma}

\begin{proof}
By induction on the derivation of the concrete edge
\(s\xrightarrow{\ S_b\ }s'\), with cases matching the abstract effects in
\cref{subsec:abs-effect}.

\textbf{Sequencing.}
For \(S;T\), the concrete execution first reaches an intermediate state
\(s_1\), and the abstract execution first reaches
\((\rho_1,C_{\texttt{ctrl}}^1)\). By the induction hypothesis for \(S\),
$s_1\Delta(\rho_1,C_{\texttt{ctrl}}^1)$.
Applying the induction hypothesis again to \(T\) gives the desired result.

\textbf{Assignment.}
For \(\texttt{x := }e\), the concrete effect updates only \(\texttt{x}\).
Let the concrete evaluation of \(e\) produce value \(v\) and dependency set
\(S_D\). By \cref{lem:expr-transfer-soundness},
\[
(v,S_D)\Delta([\![e]\!]^{\sharp}_{\rho},C_{\texttt{ctrl}}).
\]
The abstract effect updates
$\rho'=\rho\oplus\{\texttt{x}\mapsto[\![e]\!]^{\sharp}_{\rho}\}$
and leaves \(C_{\texttt{ctrl}}\) unchanged. Therefore the updated
\(\texttt{x}\) component is described by \(\rho'(\texttt{x})\), and all other
variables remain described because they are unchanged.

\textbf{Return.}
The case \(\texttt{return }e\) is identical to assignment, except that the
updated component is the distinguished variable \(\texttt{return}\).

\textbf{Conditional.}
Consider
$S_{\texttt{cond}}\equiv
\texttt{if (cond) then }S_t\texttt{ else }S_e$.
The concrete execution follows one branch, while the abstract semantics
analyzes both branches from the same input state. By the induction hypothesis,
the concrete post-state of either branch is described by the corresponding
abstract branch output:
$(\rho_t,C_{\texttt{ctrl}}^t)
\text{ or }
(\rho_e,C_{\texttt{ctrl}}^e)$.
The environment join
\(\rho_t\sqcup_{\mathrm{env}}\rho_e\) is sound: If the two branch abstract
values are provably equal by \(\mathsf{Eq}_{\mathrm A}\), the join keeps that
value, which describes both branch results; otherwise, the join uses dependence form with the union of the two branch supports, which describes either branch
result conservatively.

It remains to account for the dependency through the branch choice itself.
When the observable post-conditional abstract results may differ, the transfer
rule adds
$\{\texttt{m}\in\mathbf{M}\mid
\texttt{m}@0\in\mathbf{Sym}([\![\texttt{cond}]\!]^{\sharp}_{\rho})\}$
to \(C_{\texttt{ctrl}}'\). By
\cref{lem:expr-transfer-soundness}, this set contains every measurement-bound
variable that can affect the concrete value of \texttt{cond}. Hence every
measurement dependency introduced through the concrete branch choice is recorded
in \(C_{\texttt{ctrl}}'\). If the abstract branch results are provably equal,
then the branch choice does not change the described post-state, and no new
control dependency is needed.

Therefore, the concrete post-state is described by
$(\rho_t\sqcup_{\mathrm{env}}\rho_e,C_{\texttt{ctrl}}')$.

\textbf{Common-return lookahead.}
For the bounded-lookahead rule, the proof is the same as the conditional case,
but the equality test is applied to the two abstract return values
\(R_t\) and \(R_e\). If they differ, the measurement dependencies of
\texttt{cond} are added to \(C_{\texttt{ctrl}}'\); otherwise the two branches
produce the same described return. Thus the abstract effect describes the
concrete effect in both cases.
\qed
\end{proof}

\begin{theorem}[Soundness of abstract execution]
\label{thm:abstract-exec-soundness}
Let
\[
(\rho_0,\mathbf{Dom},\mathbf{B})
=
\textsc{InitAbsState}(H,\mathbf{M})
\]
and suppose a well-formed initial concrete state \(s_0\) satisfies
$s_0\Delta(\rho_0,\varnothing)$.
If the concrete execution of \(H\) yields
$s_0\xrightarrow{\ H\ }s^\star$
and the abstract execution yields
$(\rho_0,\varnothing)
\xRightarrow{\ H\ }
(\rho^\star,C_{\texttt{ctrl}}^\star)$,
then
$s^\star\Delta(\rho^\star,C_{\texttt{ctrl}}^\star)$.

\end{theorem}

\begin{proof}
By repeated application of \cref{lem:edge-effect-soundness} along the execution
of \(H\). The initial relation holds by construction of
\textsc{InitAbsState}: each program variable \(\texttt{x}\) is initialized to
\(\langle\texttt{x}@0\rangle\), and each measurement-bound variable
\(\texttt{m}\in\mathbf{M}\) is initialized as a binary input with bound
\([0,1]\).
\qed
\end{proof}

\begin{theorem}[Soundness of \textsc{find\_nc}]
\label{thm:find-nc-soundness}
Let
$\mathbf{M}_{\mathrm{nc}}=\textsc{find\_nc}(H,\mu)$.
For every \(\texttt{m}\in\mathbf{M}_{\mathrm{nc}}\), the initial measurement
outcome \(\texttt{m}@0\) is semantically non-contributory to the host return.
\end{theorem}
\begin{proof}
Let \(s^\star=(\nu,\sigma^\star,\delta^\star)\) be the final concrete state and
let
$(\rho^\star,C_{\texttt{ctrl}}^\star)$
be the final abstract state computed by \cref{alg:find-nc}. By \cref{thm:abstract-exec-soundness},
$s^\star\Delta(\rho^\star,C_{\texttt{ctrl}}^\star)$.
By definition of \(\Delta\) at \(\texttt{return}\), we get
$\delta^\star(\texttt{return})
\subseteq
\mathbf{MSym}_{\mathbf{M}}(\rho^\star(\texttt{return}))
\cup
C_{\texttt{ctrl}}^\star$.
By the definition of \(\mathbf{Obs}\) in \cref{subsec:detection},
\[
\mathbf{MSym}_{\mathbf{M}}(\rho^\star(\texttt{return}))
\cup
C_{\texttt{ctrl}}^\star
=
\{\texttt{m}\in\mathbf{M}\mid
\texttt{m}@0\in
\mathbf{Obs}(\rho^\star,C_{\texttt{ctrl}}^\star)\}.
\]
Therefore, if
\(\texttt{m}\in\mathbf{M}_{\mathrm{nc}}\), then
\(\texttt{m}\notin\delta^\star(\texttt{return})\). Since
\(\delta^\star(\texttt{return})\) is the exact concrete measurement-dependency
set of the host return, changing only \(\texttt{m}@0\) cannot change the
concrete return value. Hence \(\texttt{m}\) is non-contributory.
\qed
\end{proof}
\subsection{Proof details in \cref{subsec:asymptotic-analysis}}\label{app:asym-proofs}
\noindent\textbf{\Cref{theorem:complexity}.} In the worst case, the proposed host-side static analysis runs in 
$\mathcal{O}\!\left(n\cdot\mathcal{N}^2 + n\cdot|\mathcal{V}|\cdot|D|\right)$.

\begin{proof}
Polynomial addition/subtraction costs $\mathcal{O}(\mathcal{N})$ by scanning monomials once; polynomial multiplication costs $\mathcal{O}(\mathcal{N}^2)$ due to pairwise monomial products; and each \texttt{int}$(\cdot)$ requires $\mathcal{O}(\mathcal{N})$ for integer checks and interval folding.
A \texttt{random}$(\cdot,\cdot)$ is $\mathcal{O}(1)$, introducing a single fresh symbol.
Each conditional requires $\sqcup_{\mathrm{env}}$, which compares all variables in $\mathcal{V}$ and, when they differ, unions two dependence sets of size at most $|D|$, costing $\mathcal{O}(|\mathcal{V}|\cdot|D|)$.
For the bounded-lookahead rule on \(S_{\texttt{if-ret}}\), the common return expression is evaluated once under each branch environment. This adds only a constant-factor overhead relative to the branch analyses already accounted for, and therefore does not change the asymptotic bound.
Operators falling back to dependence form (e.g., division, modulo, unknown calls) also require set unions and therefore cost $\mathcal{O}(|D|)$.
Let $n_+, n_\times, n_{\texttt{int}}, n_{\texttt{r}}, n_{\texttt{c}}, n_{D}$ denote the counts of addition, multiplication, \texttt{int}, \texttt{random}, conditional, and dependence-form operators.  
Since $n_+,n_\times,n_{\texttt{int}},n_{\texttt r},n_{\texttt c},n_D \le n$, the total is
$\mathcal{O}\!\big(
n_+\mathcal{N} + n_\times \mathcal{N}^2 + n_{\texttt{int}}\mathcal{N} + n_{\texttt{r}}
+ n_{\texttt{c}}|\mathcal{V}|\cdot|D| + n_{D}|D|
\big) 
= \mathcal{O}\!\left(n\cdot\mathcal{N}^2 + n\cdot|\mathcal{V}|\cdot|D|\right)$.

\end{proof}

\noindent\textbf{\Cref{corollary:efficiency}.} When $\max\big(\ \mathcal{N}, |\mathcal{V}|, |D|\ \big)$ is bounded by a constant value, the worst-case time complexity of the static analysis on the host program is $\mathcal{O} \big(n \big)$.
\begin{proof}
    It follows directly from \cref{theorem:complexity}.
\end{proof}

\noindent\textbf{\Cref{corollary:efficiency-workflow}.} Let $G$ be the number of gates in the input quantum circuit. The worst-case time complexity of the workflow in \cref{alg:host-aware-simpl} is $\mathcal{O}(n\cdot\mathcal{N}^2 +n\cdot|\mathcal{V}|\cdot|D| + G^2)$, and when  $\max\big(\ \mathcal{N}, |\mathcal{V}|, |D|\ \big)$ is bounded by a constant value, the complexity becomes $\mathcal{O}(n + G^2)$.
\begin{proof}
    It directly follows \cref{theorem:complexity} and Theorem $4$ in \cite{10.1007/978-3-031-97632-2_10}.
\end{proof}
\subsection{Proof details in \cref{subsec:method-gpu-acc}}\label{app:gpu-proofs}
During lowering, we maintain a representative map
\[
\kappa:\mathcal V\cup\{\texttt{return}\}\to
\widehat{\mathcal V}\cup\{\texttt{return}\},
\]
which maps each source-level variable to its current SSA representative in the
lowered program.  The initial map \(\kappa_0\) maps each source variable to
its incoming SSA version and maps \texttt{return} to \texttt{return}.  For the
whole program, we write
\[
(\widehat H,\kappa^\star)=\textsc{lowerBlock}(H,\kappa_0),
\]
where \(\kappa^\star\) is the final representative map produced by the
lowering.

Given a lowered abstract environment \(\widehat{\rho}\), its projection back
to source-level variables under a representative map \(\kappa\) is
\[
\pi_\kappa(\widehat{\rho})(\texttt{x})
=
\widehat{\rho}(\kappa(\texttt{x}))
\qquad
\text{for all }
\texttt{x}\in\mathcal V\cup\{\texttt{return}\}.
\]
This projection forgets lowering-introduced temporaries and keeps only the
abstract values of source-level representatives.

Let
$(\rho_0,\mathbf{Dom},\mathbf{B})
=
\textsc{InitAbsState}(H,\mathbf{M})$.
The initial lowered environment \(\widehat{\rho}_0\) is defined by
\[
\widehat{\rho}_0(\kappa_0(\texttt{x}))=\rho_0(\texttt{x})
\qquad
\text{for all }
\texttt{x}\in\mathcal V\cup\{\texttt{return}\}.
\]
The maps \(\mathbf{Dom}\) and \(\mathbf{B}\) are shared by the source and
lowered analyses.

The following proposition shows that, after projecting the lowered environment back to
source-level variables, the lowered analysis produces the same abstract
\texttt{return} value and the same control-dependency set
\(C_{\texttt{ctrl}}\), and therefore the same output of \textsc{find\_nc} (\cref{alg:find-nc}).

\begin{proposition}[Preservation under SSA lowering]
\label{prop:lowering-preserves-find-nc}
Let
$(\widehat H,\kappa^\star)=\textsc{lowerBlock}(H,\kappa_0)$
be the levelized SSA lowering of \(H\).  Suppose
$(\rho^\star,C^\star_{\texttt{ctrl}})
=
\textsc{AnalyzeBlock}(H,\rho_0,\varnothing,\mathbf{M})$
and
$(\widehat{\rho}^\star,\widehat{C}^\star_{\texttt{ctrl}})
=
\textsc{AnalyzeBlock}(\widehat H,\widehat{\rho}_0,\varnothing,\mathbf{M})$.
Then
$\pi_{\kappa^\star}(\widehat{\rho}^\star)=\rho^\star$
and
$\widehat{C}^\star_{\texttt{ctrl}}=C^\star_{\texttt{ctrl}}$.
\end{proposition}

\begin{proof}
We prove a simulation invariant for the representative map maintained by the lowering.  
Consider any source block \(S_b\) and write
\[
(\widehat{S_b},\kappa_{\mathrm{out}})
=
\textsc{lowerBlock}(S_b,\kappa_{\mathrm{in}}).
\]
Assume the input states agree under the entry representative map:
\[
\pi_{\kappa_{\mathrm{in}}}(\widehat{\rho}_{\mathrm{in}})
=
\rho_{\mathrm{in}}
\qquad
\text{and}
\qquad
\widehat{C}^{\mathrm{in}}_{\texttt{ctrl}}
=
C^{\mathrm{in}}_{\texttt{ctrl}}.
\]
We show, by structural induction on \(S_b\), that the output states agree
under the exit representative map:
\[
\pi_{\kappa_{\mathrm{out}}}(\widehat{\rho}_{\mathrm{out}})
=
\rho_{\mathrm{out}}
\qquad
\text{and}
\qquad
\widehat{C}^{\mathrm{out}}_{\texttt{ctrl}}
=
C^{\mathrm{out}}_{\texttt{ctrl}}.
\]

For an assignment \(\texttt{x}:=e\), the lowering decomposes \(e\) into an
acyclic sequence of primitive SSA instructions.  Let \(v_e\) be the final
destination of this sequence.  The lowering updates the representative map by
setting
$\kappa_{\mathrm{out}}(\texttt{x})=v_e$
and leaving the representatives of all other source variables unchanged.  By
induction on the syntax of \(e\), the lowered instruction sequence computes at
\(v_e\) exactly the same abstract value as the source abstract evaluation of
\(e\).  Fresh temporaries introduced during this decomposition do not appear
in the projection.  Hence the projected lowered environment agrees with the
source environment after the assignment.  The case of
\(\texttt{return }e\) is identical, except that the final destination becomes
the representative of \texttt{return}.

Sequencing follows immediately by applying the induction hypothesis to the
first statement and then to the second.

For a structured conditional, the lowering recursively lowers both branches
from the same entry representative map, producing branch representative maps
for the then- and else-branches.  For each source-level live-out variable, the
lowering either reuses a common branch representative or introduces an
explicit merge
$u_{\texttt{x}}
:=
\phi_{\mathrm{if}}(c,v_{\texttt{x}}^t,v_{\texttt{x}}^e)$.
The abstract transfer of this merge is defined by the same branch-join rule as
the source analysis: the merged value is
$\widehat{\rho}_t(v_{\texttt{x}}^t)
\sqcup_{\mathrm A}
\widehat{\rho}_e(v_{\texttt{x}}^e)$,
and the same condition-dependent measurements are added to
\(C_{\texttt{ctrl}}\) exactly when the source conditional rule would add them.
By the induction hypothesis, the branch output states already agree after
projection.  Therefore the lowered merges reproduce the source environment
join \(\sqcup_{\mathrm{env}}\) and the same control-dependency update.  Hence
the invariant is preserved across conditionals.

It remains to justify the levelized GPU schedule.  By the levelization
property, every instruction in a level reads only incoming SSA values or
values defined in earlier levels, and destinations inside the same level are
distinct.  Thus same-level instructions do not interfere through reads or
writes.  The only shared update is the enlargement of
\(C_{\texttt{ctrl}}\), which is a set union and is independent of the
intra-level execution order.  Therefore the level-wise GPU execution computes
the same lowered abstract state as a sequential execution of \(\widehat H\).

Applying the invariant to the whole program gives
$\pi_{\kappa^\star}(\widehat{\rho}^\star)=\rho^\star$
and
$\widehat{C}^\star_{\texttt{ctrl}}=C^\star_{\texttt{ctrl}}$.
\qed
\end{proof}

\clearpage

\bibliographystyle{splncs04}

\begin{thebibliography}{10}
\providecommand{\url}[1]{\texttt{#1}}
\providecommand{\urlprefix}{URL }
\providecommand{\doi}[1]{https://doi.org/#1}

\bibitem{QiskitNatureFreezeCore}
{Qiskit Nature}: Freezecoretransformer.
  \url{https://qiskit-community.github.io/qiskit-nature/stubs/qiskit_nature.second_q.transformers.FreezeCoreTransformer.html}
  (2024)

\bibitem{Abedi2023quantumlazytraining}
Abedi, E., Beigi, S., Taghavi, L.: Quantum {L}azy {T}raining. {Quantum}
  \textbf{7}, ~989 (Apr 2023). \doi{10.22331/q-2023-04-27-989},
  \url{https://doi.org/10.22331/q-2023-04-27-989}

\bibitem{alfred2007compilers}
Aho, A.V., Lam, M.S., Sethi, R., Ullman, J.D.: Compilers: Principles,
  Techniques, and Tools. Pearson, 2 edn. (2007)

\bibitem{barthe2014ssa}
Barthe, G., Demange, D., Pichardie, D.: Formal verification of an {SSA}-based
  middle-end for {CompCert}. ACM Transactions on Programming Languages and
  Systems  \textbf{36}(1) (2014). \doi{10.1145/2579080}

\bibitem{bergholm2022pennylaneautomaticdifferentiationhybrid}
Bergholm, V., Izaac, J., Schuld, M., Gogolin, C., Ahmed, S., Ajith, V., Alam,
  M.S., Alonso-Linaje, G., AkashNarayanan, B., Asadi, A., Arrazola, J.M., Azad,
  U., Banning, S., Blank, C., Bromley, T.R., Cordier, B.A., Ceroni, J.,
  Delgado, A., Matteo, O.D., Dusko, A., Garg, T., Guala, D., Hayes, A., Hill,
  R., Ijaz, A., Isacsson, T., Ittah, D., Jahangiri, S., Jain, P., Jiang, E.,
  Khandelwal, A., Kottmann, K., Lang, R.A., Lee, C., Loke, T., Lowe, A.,
  McKiernan, K., Meyer, J.J., Montañez-Barrera, J.A., Moyard, R., Niu, Z.,
  O'Riordan, L.J., Oud, S., Panigrahi, A., Park, C.Y., Polatajko, D., Quesada,
  N., Roberts, C., Sá, N., Schoch, I., Shi, B., Shu, S., Sim, S., Singh, A.,
  Strandberg, I., Soni, J., Száva, A., Thabet, S., Vargas-Hernández, R.A.,
  Vincent, T., Vitucci, N., Weber, M., Wierichs, D., Wiersema, R., Willmann,
  M., Wong, V., Zhang, S., Killoran, N.: Pennylane: Automatic differentiation
  of hybrid quantum-classical computations (2022),
  \url{https://arxiv.org/abs/1811.04968}

\bibitem{10313887}
Bermot, E., Zoufal, C., Grossi, M., Schuhmacher, J., Tacchino, F., Vallecorsa,
  S., Tavernelli, I.: { Quantum Generative Adversarial Networks For Anomaly
  Detection In High Energy Physics }. In: 2023 IEEE International Conference on
  Quantum Computing and Engineering (QCE). pp. 331--341. IEEE Computer Society,
  Los Alamitos, CA, USA (Sep 2023). \doi{10.1109/QCE57702.2023.00045},
  \url{https://doi.ieeecomputersociety.org/10.1109/QCE57702.2023.00045}

\bibitem{bouland2020prospectschallengesquantumfinance}
Bouland, A., van Dam, W., Joorati, H., Kerenidis, I., Prakash, A.: Prospects
  and challenges of quantum finance (2020),
  \url{https://arxiv.org/abs/2011.06492}

\bibitem{Bravyi2017Tapering}
Bravyi, S., Gambetta, J.M., Mezzacapo, A., Temme, K.: Tapering off qubits to
  simulate fermionic hamiltonians. arXiv preprint arXiv:1701.08213  (2017).
  \doi{10.48550/arXiv.1701.08213}

\bibitem{broughton2021tensorflowquantumsoftwareframework}
Broughton, M., Verdon, G., McCourt, T., Martinez, A.J., Yoo, J.H., Isakov,
  S.V., Massey, P., Halavati, R., Niu, M.Y., Zlokapa, A., Peters, E., Lockwood,
  O., Skolik, A., Jerbi, S., Dunjko, V., Leib, M., Streif, M., Dollen, D.V.,
  Chen, H., Cao, S., Wiersema, R., Huang, H.Y., McClean, J.R., Babbush, R.,
  Boixo, S., Bacon, D., Ho, A.K., Neven, H., Mohseni, M.: Tensorflow quantum: A
  software framework for quantum machine learning (2021),
  \url{https://arxiv.org/abs/2003.02989}

\bibitem{Chen_2022_Partial_Eq_Check}
Chen, T.F., Jiang, J.H.R., Hsieh, M.H.: Partial equivalence checking of quantum
  circuits. In: 2022 IEEE International Conference on Quantum Computing and
  Engineering (QCE). p. 594–604. IEEE (Sep 2022).
  \doi{10.1109/qce53715.2022.00082},
  \url{http://dx.doi.org/10.1109/QCE53715.2022.00082}

\bibitem{10.1007/978-3-031-97632-2_10}
Chen, Y., Mendl, C.B., Seidl, H.: Dead gate elimination. In: Lees, M.H., Cai,
  W., Cheong, S.A., Su, Y., Abramson, D., Dongarra, J.J., Sloot, P.M.A. (eds.)
  Computational Science -- ICCS 2025. pp. 135--150. Springer Nature
  Switzerland, Cham (2025)

\bibitem{chen_QCP_2023}
Chen, Y., Stade, Y.: Quantum constant propagation. In: Hermenegildo, M.V.,
  Morales, J.F. (eds.) Static Analysis. pp. 164--189. Springer Nature
  Switzerland, Cham (2023),
  \url{https://link.springer.com/chapter/10.1007/978-3-031-44245-2_9}

\bibitem{Cortes1995SVM}
Cortes, C., Vapnik, V.: Support-vector networks. Machine Learning  \textbf{20},
   273--297 (1995). \doi{10.1007/BF00994018}

\bibitem{cytron1991ssa}
Cytron, R., Ferrante, J., Rosen, B.K., Wegman, M.N., Zadeck, F.K.: Efficiently
  computing static single assignment form and the control dependence graph. ACM
  Transactions on Programming Languages and Systems  \textbf{13}(4),  451--490
  (1991). \doi{10.1145/115372.115320}

\bibitem{van2019electronic}
van Dijk, J.P., Charbon, E., Sebastiano, F.: The electronic interface for
  quantum processors. Microprocessors and Microsystems  \textbf{66},  90--101
  (2019)

\bibitem{Egger2021CreditRisk}
Egger, D.J., Guti{\'e}rrez, R.G., Mestre, J.C., Woerner, S.: Credit risk
  analysis using quantum computers. IEEE Transactions on Computers
  \textbf{70}(12),  2136--2145 (2021). \doi{10.1109/TC.2020.3038063}

\bibitem{Farhi2014QAOA}
Farhi, E., Goldstone, J., Gutmann, S.: A quantum approximate optimization
  algorithm  (2014), \url{https://arxiv.org/abs/1411.4028}

\bibitem{Glover2019QUBO}
Glover, F., Kochenberger, G., Du, Y.: Quantum bridge analytics {I}: A tutorial
  on formulating and using {QUBO} models. 4OR  \textbf{17},  335--371 (2019).
  \doi{10.1007/s10288-019-00424-y}

\bibitem{Havlicek2019QuantumFeatures}
Havl{\'i}{\v c}ek, V., C{\'o}rcoles, A.D., Temme, K., Harrow, A.W., Kandala,
  A., Chow, J.M., Gambetta, J.M.: Supervised learning with quantum-enhanced
  feature spaces. Nature  \textbf{567},  209--212 (2019).
  \doi{10.1038/s41586-019-0980-2}

\bibitem{jojo2024quantumalgorithmstensorsvd}
Jojo, J., Khandelwal, A., Chandra, M.G.: Quantum algorithms for tensor-svd
  (2024), \url{https://arxiv.org/abs/2405.19485}

\bibitem{Kaul_2023}
Kaul, M., Küchler, A., Banse, C.: A uniform representation of classical and
  quantum source code for static code analysis. In: 2023 IEEE International
  Conference on Quantum Computing and Engineering (QCE). p. 1013–1019. IEEE
  (Sep 2023). \doi{10.1109/qce57702.2023.00115},
  \url{http://dx.doi.org/10.1109/QCE57702.2023.00115}

\bibitem{khammassi2022scalable}
Khammassi, N., Morris, R.W., Premaratne, S., Luthi, F., Borjans, F., Suzuki,
  S., Flory, R., Ibarra, L.P.O., Lampert, L., Matsuura, A.Y.: A scalable
  microarchitecture for efficient instruction-driven signal synthesis and
  coherent qubit control. arXiv preprint arXiv:2205.06851  (2022)

\bibitem{kildall1973unified}
Kildall, G.A.: A unified approach to global program optimization. In:
  Proceedings of the 1st annual ACM SIGACT-SIGPLAN symposium on Principles of
  programming languages. pp. 194--206 (1973)

\bibitem{kissinger2019pyzx}
Kissinger, A., van~de Wetering, J.: Pyzx: Large scale automated diagrammatic
  reasoning. arXiv preprint arXiv:1904.04735  (2019)

\bibitem{LaRose2022mitiqsoftware}
LaRose, R., Mari, A., Kaiser, S., Karalekas, P.J., Alves, A.A., Czarnik, P.,
  El~Mandouh, M., Gordon, M.H., Hindy, Y., Robertson, A., Thakre, P., Wahl, M.,
  Samuel, D., Mistri, R., Tremblay, M., Gardner, N., Stemen, N.T., Shammah, N.,
  Zeng, W.J.: Mitiq: {A} software package for error mitigation on noisy quantum
  computers. {Quantum}  \textbf{6}, ~774 (Aug 2022).
  \doi{10.22331/q-2022-08-11-774},
  \url{https://doi.org/10.22331/q-2022-08-11-774}

\bibitem{lemerre2023ssa}
Lemerre, M.: {SSA} translation is an abstract interpretation. Proceedings of
  the ACM on Programming Languages  \textbf{7}(POPL),  1895--1924 (2023).
  \doi{10.1145/3571258}

\bibitem{Li_2024}
Li, P., Liu, J., Gonzales, A., Saleem, Z.H., Zhou, H., Hovland, P.: Qutracer:
  Mitigating quantum gate and measurement errors by tracing subsets of qubits.
  In: 2024 ACM/IEEE 51st Annual International Symposium on Computer
  Architecture (ISCA). p. 103–117. IEEE (Jun 2024).
  \doi{10.1109/isca59077.2024.00018},
  \url{http://dx.doi.org/10.1109/ISCA59077.2024.00018}

\bibitem{maciejewski2020mitigation}
Maciejewski, F.B., Zimbor{\'a}s, Z., Oszmaniec, M.: Mitigation of readout noise
  in near-term quantum devices by classical post-processing based on detector
  tomography. Quantum  \textbf{4}, ~257 (2020)

\bibitem{10313603}
Matondo-Mvula, N., Elleithy, K.: Advances in quantum medical image analysis
  using machine learning: Current status and future directions. In: 2023 IEEE
  International Conference on Quantum Computing and Engineering (QCE). vol.~01,
  pp. 367--377 (2023). \doi{10.1109/QCE57702.2023.00049}

\bibitem{McArdle2020QChem}
McArdle, S., Endo, S., Aspuru-Guzik, A., Benjamin, S.C., Yuan, X.: Quantum
  computational chemistry. Reviews of Modern Physics  \textbf{92}(1),  015003
  (2020). \doi{10.1103/RevModPhys.92.015003}

\bibitem{8638598}
McCaskey, A., Dumitrescu, E., Liakh, D., Humble, T.: Hybrid programming for
  near-term quantum computing systems. In: 2018 IEEE International Conference
  on Rebooting Computing (ICRC). pp. 1--12 (2018).
  \doi{10.1109/ICRC.2018.8638598}

\bibitem{McClean2016VQE}
McClean, J.R., Romero, J., Babbush, R., Aspuru-Guzik, A.: The theory of
  variational hybrid quantum-classical algorithms. New Journal of Physics
  \textbf{18}(2),  023023 (2016). \doi{10.1088/1367-2630/18/2/023023}

\bibitem{spa}
M\o{}ller, A., Schwartzbach, M.I.: Static program analysis (October 2018),
  department of Computer Science, Aarhus University,
  \texttt{http://cs.au.dk/\~{}amoeller/spa/}

\bibitem{Montanaro2015MonteCarlo}
Montanaro, A.: Quantum speedup of monte carlo methods. Proceedings of the Royal
  Society A  \textbf{471}(2181),  20150301 (2015). \doi{10.1098/rspa.2015.0301}

\bibitem{Motta2021LowRank}
Motta, M., Ye, E., McClean, J.R., Li, Z., Minnich, A.J., Babbush, R., Chan,
  G.K.L.: Low rank representations for quantum simulation of electronic
  structure. npj Quantum Information  \textbf{7}, ~83 (2021).
  \doi{10.1038/s41534-021-00416-z}

\bibitem{nachman2020unfolding}
Nachman, B., Urbanek, M., de~Jong, W.A., Bauer, C.W.: Unfolding quantum
  computer readout noise. npj Quantum Information  \textbf{6}(1), ~84 (2020)

\bibitem{bookFlemming1999}
Nielson, F., Nielson, H., Hankin, C.: Principles of Program Analysis (01 1999).
  \doi{10.1007/978-3-662-03811-6}

\bibitem{cuda_best_practices}
{NVIDIA}: Cuda c++ best practices guide.
  \url{https://docs.nvidia.com/cuda/cuda-c-best-practices-guide/} (2026),
  accessed 2026-04-14

\bibitem{cuda_programming_guide}
{NVIDIA}: Cuda c++ programming guide.
  \url{https://docs.nvidia.com/cuda/cuda-programming-guide/} (2026), accessed
  2026-04-14

\bibitem{Paltenghi_2024}
Paltenghi, M., Pradel, M.: Analyzing quantum programs with lintq: A static
  analysis framework for qiskit. Proceedings of the ACM on Software Engineering
   \textbf{1}(FSE),  2144–2166 (Jul 2024). \doi{10.1145/3660802},
  \url{http://dx.doi.org/10.1145/3660802}

\bibitem{Peruzzo2014VQE}
Peruzzo, A., McClean, J., Shadbolt, P., Yung, M.H., Zhou, X.Q., Love, P.J.,
  Aspuru-Guzik, A., O'Brien, J.L.: A variational eigenvalue solver on a
  photonic quantum processor. Nature Communications  \textbf{5}, ~4213 (2014).
  \doi{10.1038/ncomms5213}

\bibitem{Preskill_2018}
Preskill, J.: Quantum computing in the nisq era and beyond. Quantum
  \textbf{2}, ~79 (Aug 2018). \doi{10.22331/q-2018-08-06-79},
  \url{http://dx.doi.org/10.22331/q-2018-08-06-79}

\bibitem{Péter_1954}
Péter, R.: H. g. rice. classes of recursively enumerable sets and their
  decision problems. transactions of the american mathematical society, vol. 74
  (1953) pp. 358–366. Journal of Symbolic Logic  \textbf{19}(2),  121–122
  (1954). \doi{10.2307/2268870}

\bibitem{Qiskit}
{Qiskit contributors}: Qiskit: An open-source framework for quantum computing
  (2023). \doi{10.5281/zenodo.2573505}

\bibitem{Quetschlich2023MQTBench}
Quetschlich, N., Burgholzer, L., Wille, R.: {MQT Bench}: Benchmarking software
  and design automation tools for quantum computing. Quantum  \textbf{7}, ~1062
  (2023). \doi{10.22331/q-2023-07-20-1062}

\bibitem{quetschlich2024equivalencecheckingclassicalcircuits}
Quetschlich, N., Forster, T., Osterwind, A., Helms, D., Wille, R.: Towards
  equivalence checking of classical circuits using quantum computing (2024),
  \url{https://arxiv.org/abs/2408.14539}

\bibitem{remme2025optimizationhybridquantumclassicalalgorithms}
Remme, L., Weinert, A., Waschk, A.: Optimization of hybrid quantum-classical
  algorithms (2025), \url{https://arxiv.org/abs/2505.12853}

\bibitem{Schuld2019FeatureHilbert}
Schuld, M., Killoran, N.: Quantum machine learning in feature hilbert spaces.
  Physical Review Letters  \textbf{122},  040504 (2019).
  \doi{10.1103/PhysRevLett.122.040504}

\bibitem{seidl2012compiler}
Seidl, H., Wilhelm, R., Hack, S.: Compiler Design: Analysis and Transformation.
  Springer (2012)

\bibitem{Sivarajah_tket_2021}
Sivarajah, S., Dilkes, S., Cowtan, A., Simmons, W., Edgington, A., Duncan, R.:
  t|ket⟩: a retargetable compiler for nisq devices. Quantum Science and
  Technology  \textbf{6}(1),  014003 (nov 2020).
  \doi{10.1088/2058-9565/ab8e92},
  \url{https://dx.doi.org/10.1088/2058-9565/ab8e92}

\bibitem{Stamatopoulos2020Option}
Stamatopoulos, N., Egger, D.J., Sun, Y., Zoufal, C., Iten, R., Shen, N.,
  Woerner, S.: Option pricing using quantum computers. Quantum  \textbf{4},
  ~291 (2020). \doi{10.22331/q-2020-07-06-291}

\bibitem{Steiger2018projectqopensource}
Steiger, D.S., H{\"{a}}ner, T., Troyer, M.: Project{Q}: an open source software
  framework for quantum computing. {Quantum}  \textbf{2}, ~49 (Jan 2018).
  \doi{10.22331/q-2018-01-31-49},
  \url{https://doi.org/10.22331/q-2018-01-31-49}

\bibitem{Tibshirani1996Lasso}
Tibshirani, R.: Regression shrinkage and selection via the lasso. Journal of
  the Royal Statistical Society: Series B  \textbf{58}(1),  267--288 (1996).
  \doi{10.1111/j.2517-6161.1996.tb02080.x}

\bibitem{Verteletskyi2020Measurement}
Verteletskyi, V., Yen, T.C., Izmaylov, A.F.: Measurement optimization in the
  variational quantum eigensolver using a minimum clique cover. The Journal of
  Chemical Physics  \textbf{152}(12),  124114 (2020). \doi{10.1063/1.5141458}

\bibitem{Woerner2019Risk}
Woerner, S., Egger, D.J.: Quantum risk analysis. npj Quantum Information
  \textbf{5}, ~15 (2019). \doi{10.1038/s41534-019-0130-6}

\bibitem{10.1145/2499370.2462164}
Zhao, J., Nagarakatte, S., Martin, M.M., Zdancewic, S.: Formal verification of
  ssa-based optimizations for llvm. SIGPLAN Not.  \textbf{48}(6),  175–186
  (Jun 2013). \doi{10.1145/2499370.2462164},
  \url{https://doi.org/10.1145/2499370.2462164}

\bibitem{zhao2023qcheckerdetectingbugsquantum}
Zhao, P., Wu, X., Li, Z., Zhao, J.: Qchecker: Detecting bugs in quantum
  programs via static analysis (2023), \url{https://arxiv.org/abs/2304.04387}

\end{thebibliography}

\end{document}